\documentclass{article}

\usepackage[preprint,nonatbib]{nips_2018}

\usepackage[utf8]{inputenc} 
\usepackage[T1]{fontenc}    
\usepackage{hyperref}       
\usepackage{url}            
\usepackage{booktabs}       
\usepackage{amsfonts}       
\usepackage{nicefrac}       
\usepackage{microtype}      

\usepackage{subcaption}
\usepackage{amsmath, amssymb, bm, epsfig, psfrag,amsthm}
\usepackage{algorithm,algorithmic}
\usepackage{bbm}
\usepackage[usenames]{color}
\usepackage{tikz}
\usepackage{etoolbox}
\usepackage{enumerate}
\usetikzlibrary{shapes,arrows}
\usetikzlibrary{patterns}
\usepackage{pgfplots}
\usepackage{url,hyperref}
\usepackage{xspace}
\usepackage{enumerate}
\usepackage{mathtools}
\usepackage{epstopdf}

\newtoggle{conference}
\toggletrue{conference} 

\usepackage[normalem]{ulem} 

\renewcommand{\tilde}{\widetilde}
\renewcommand{\hat}{\widehat}

\def\beq{\begin{equation}}
\def\eeq{\end{equation}}
\def\beqa{\begin{eqnarray}}
\def\eeqa{\end{eqnarray}}
\def\beqan{\begin{eqnarray*}}
\def\eeqan{\end{eqnarray*}}

\def\R{{\mathbb{R}}}

\def\argmin{\mathop{\mathrm{arg\,min}}}
\def\argmax{\mathop{\mathrm{arg\,max}}}
\def\x{\times}

\newtheorem{definition}{Definition}
\newtheorem{theorem}{Theorem}
\newtheorem{lemma}{Lemma}

\newtheorem{assumption}{Assumption}

\def\xhat{\hat{x}}

\def\arr{\rightarrow}
\def\Exp{\mathbb{E}}

\def\alphabar{\overline{\alpha}}
\def\etabar{\overline{\eta}}
\def\gammabar{\overline{\gamma}}

\def\tm1{t\! - \! 1}
\def\tp1{t\! + \! 1}
\def\km1{k\! - \! 1}
\def\kp1{k\! + \! 1}
\def\ip1{i\! + \! 1}
\def\im1{i\! - \! 1}
\def\lm1{\ell\! - \! 1}
\def\lp1{\ell\! + \! 1}

\newcommand{\zero}{\mathbf{0}}

\newcommand{\bbf}{\mathbf{b}}
\newcommand{\bbfhat}{\hat{\mathbf{b}}}
\newcommand{\cbf}{\mathbf{c}}
\newcommand{\cbfhat}{\hat{\mathbf{c}}}

\newcommand{\fbf}{\mathbf{f}}
\newcommand{\gbf}{\mathbf{g}}
\newcommand{\hbf}{\mathbf{h}}
\newcommand{\pbf}{\mathbf{p}}

\newcommand{\qbf}{\mathbf{q}}

\newcommand{\rbf}{\mathbf{r}}

\newcommand{\sbf}{\mathbf{s}}

\newcommand{\ubf}{\mathbf{u}}
\newcommand{\ubfhat}{\hat{\mathbf{u}}}

\newcommand{\vbf}{\mathbf{v}}

\newcommand{\wbf}{\mathbf{w}}

\newcommand{\xbf}{\mathbf{x}}
\newcommand{\xbfhat}{\hat{\mathbf{x}}}

\newcommand{\ybf}{\mathbf{y}}

\newcommand{\zbf}{\mathbf{z}}

\newcommand{\Abf}{\mathbf{A}}
\newcommand{\Bbf}{\mathbf{B}}

\newcommand{\Dbf}{\mathbf{D}}

\newcommand{\Hbf}{\mathbf{H}}

\newcommand{\Ibf}{\mathbf{I}}
\newcommand{\Jbf}{\mathbf{J}}

\newcommand{\Pbf}{\mathbf{P}}

\newcommand{\Qbf}{\mathbf{Q}}
\newcommand{\Rbf}{\mathbf{R}}

\newcommand{\Sbf}{\mathbf{S}}
\newcommand{\Ubf}{\mathbf{U}}

\newcommand{\Vbf}{\mathbf{V}}

\newcommand{\Xbf}{\mathbf{X}}

\def\etabf{{\boldsymbol \eta}}

\def\xibf{{\boldsymbol \xi}}

\def\Phibf{{\boldsymbol{\Phi}}}
\def\Psibf{{\boldsymbol{\Psi}}}

\newcommand{\phibf}{{\bm{\phi}}}

\newcommand{\tran}{^{\text{\sf T}}}

\newcommand{\defn}{:=}

\def\eqd{\stackrel{d}{=}}

\def\Norm{{\mathcal N}}
\def\Range{\mathrm{Range}}
\def\Diag{\mathrm{Diag}}
\def\alphabar{\overline{\alpha}}

\def\Ecal{{\mathcal E}}

\def\Kit{K_{\rm it}}

\def\ONinv{O(\tfrac{1}{\sqrt{N}})}

\newcommand{\bkt}[1]{{\langle #1 \rangle}}

\title{Plug-in Estimation in High-Dimensional Linear Inverse Problems:  A Rigorous Analysis}

\author{
 Alyson K. Fletcher\\
 Dept.\ Statistics\\
 UC Los Angeles\\
 \texttt{akfletcher@ucla.edu} \\
 \And
 Sundeep Rangan\\
 Dept.\ ECE\\
 NYU\\
 \texttt{srangan@nyu.edu} \\
 \And
 Subrata Sarkar\\
 Dept.\ ECE\\
 The Ohio State Univ.\\
 \texttt{sarkar.51@osu.edu}
 \And
 Philip Schniter\\
 Dept.\ ECE\\
 The Ohio State Univ.\\
 \texttt{schniter.1@osu.edu}
}

\begin{document}

\maketitle

\begin{abstract}
Estimating a vector $\xbf$ from noisy linear measurements $\Abf\xbf+\wbf$
often requires use of prior knowledge or structural constraints
on $\xbf$ for accurate reconstruction.
Several recent works have considered combining 
linear least-squares estimation with a generic or ``plug-in'' denoiser function
that can be designed in a modular manner based on the prior knowledge about $\xbf$.
While these methods have shown excellent performance, it has been difficult 
to obtain rigorous performance guarantees.
This work considers plug-in denoising combined with the
recently-developed Vector Approximate Message Passing (VAMP) algorithm,
which is itself derived via Expectation Propagation techniques.
It shown that the mean squared error of this ``plug-and-play" 
VAMP can be exactly predicted for high-dimensional right-rotationally invariant random $\Abf$ and Lipschitz denoisers.
The method is demonstrated on applications in image recovery and parametric bilinear estimation.
\end{abstract}

\section{Introduction}

The estimation of an unknown vector $\xbf^0 \in \R^N$
from noisy linear measurements $\ybf$ of the form
\beq \label{eq:yAx}
    \ybf = \Abf\xbf^0 + \wbf \in \R^M,
\eeq
where $\Abf \in \R^{M \x N}$ is a known transform and $\wbf$ is disturbance,
arises in a wide-range of learning and inverse problems.
In many high-dimensional situations, such as when the measurements are fewer than the unknown
parameters (i.e., $M\ll N$), it is essential to incorporate known structure on $\xbf^0$
in the estimation process.
A fundamental challenge is how to perform structured estimation
of $\xbf^0$ while maintaining computational efficiency and a tractable analysis.

\emph{Approximate message passing} (AMP), originally proposed in \cite{DonohoMM:09},
refers to a powerful class of algorithms that can be applied to reconstruction
of $\xbf^0$ from \eqref{eq:yAx} that can easily incorporate a wide class of
statistical priors.
In this work, we restrict our attention to $\wbf\sim\Norm(\zero,\gamma_w^{-1}\Ibf)$,
noting that AMP was extended to non-Gaussian measurements in
\cite{Rangan:11-ISIT,RanSRFC:13-ISIT,RanSchFle:14-ISIT}.
AMP is computationally efficient, in that it generates a sequence of estimates
$\{\xbfhat_k\}_{k=0}^\infty$ by iterating the steps
\begin{subequations} \label{eq:AMP}
\begin{align}
  \xbfhat_{k}
  &= \gbf(\rbf_{k},\gamma_k) \\
  \vbf_{k}
  &= \ybf - \Abf\xbfhat_{k} + \tfrac{N}{M}
     \langle \nabla\gbf(\rbf_{k-1},\gamma_{k-1})\rangle \vbf_{k-1} \\
  \rbf_{k+1}
  &= \xbfhat_k + \Abf\tran\vbf_k , \quad \gamma_{k+1} = M/\|\vbf_k\|^2 ,
\end{align}
\end{subequations}
initialized with
$\rbf_0=\Abf\tran\ybf$, $\gamma_0=M/\|\ybf\|^2$, $\vbf_{-1}=\zero$, and
assuming $\Abf$ is scaled so that $\|\Abf\|_F^2\approx N$.
In \eqref{eq:AMP}, $\gbf:\R^N \x \R \rightarrow \R^N$ is an estimation function
chosen based on prior knowledge about $\xbf^0$, and
$\langle \nabla\gbf(\rbf,\gamma) \rangle \defn \frac{1}{N}\sum_{n=1}^N \frac{\partial g_n(\rbf,\gamma)}{\partial r_n}$
denotes the divergence of $\gbf(\rbf,\gamma)$.
For example, if $\xbf^0$ is known to be sparse, then it is common to
choose $\gbf(\cdot)$ to be the componentwise soft-thresholding function,
in which case AMP iteratively solves the LASSO \cite{Tibshirani:96} problem.

Importantly, for large, i.i.d., sub-Gaussian random matrices $\Abf$
and Lipschitz denoisers $\gbf(\cdot)$, the performance of AMP
can be exactly predicted by a scalar \emph{state evolution} (SE),
which also provides testable conditions for optimality
\cite{BayatiM:11,javanmard2013state,berthier2017state}.
The initial work \cite{BayatiM:11,javanmard2013state} focused on the case where
$\gbf(\cdot)$ is a separable function with identical components
(i.e., $[\gbf(\rbf,\gamma)]_n=g(r_n,\gamma)~\forall n$),
while the later work \cite{berthier2017state} allowed non-separable $\gbf(\cdot)$.
Interestingly, these SE analyses establish the fact that
\begin{equation}
\label{eq:rk}
\rbf_k = \xbf^0 + \Norm(\zero,\Ibf/\gamma_k),
\end{equation}
leading to the important interpretation that $\gbf(\cdot)$ acts as a \emph{denoiser}.
This interpretation provides guidance on how to choose $\gbf(\cdot)$.
For example, if $\xbf$ is i.i.d.\ with a known prior, then
\eqref{eq:rk} suggests to choose
a separable $\gbf(\cdot)$ composed of minimum mean-squared error (MMSE) scalar denoisers
$g(r_n,\gamma)=\Exp(x_n|r_n=x_n+\Norm(0,1/\gamma))$.
In this case, \cite{BayatiM:11,javanmard2013state} established that,
whenever the SE has a unique fixed point, the estimates $\xbfhat_k$
generated by AMP converge to the Bayes optimal estimate of $\xbf^0$ from $\ybf$.
As another example, if $\xbf$ is a natural image,
for which an analytical prior is lacking,
then \eqref{eq:rk} suggests to choose $\gbf(\cdot)$ as a sophisticated
image-denoising algorithm like BM3D \cite{dabov2007image} 
or DnCNN \cite{zhang2017beyond},
as proposed in \cite{Metzler2016denoise}.
Many other examples of structured estimators $\gbf(\cdot)$ can be considered;
we refer the reader to \cite{berthier2017state} and Section~\ref{sec:num}.
Prior to \cite{berthier2017state}, AMP SE results were established for
special cases of $\gbf(\cdot)$ in \cite{DonohoJM:13,ma2017analysis}.
Plug-in denoisers have been combined in related algorithms
\cite{venkatakrishnan2013plug,chen2017bm3d,wang2017parameter}.

An important limitation of AMP's SE is that it holds only for large,
i.i.d., sub-Gaussian $\Abf$.
AMP itself often fails to converge with small deviations
from i.i.d.\ sub-Gaussian $\Abf$, such as when
$\Abf$ is mildly ill-conditioned or non-zero-mean
\cite{RanSchFle:14-ISIT,Caltagirone:14-ISIT,Vila:ICASSP:15}.
Recently, a robust alternative to AMP called \emph{vector AMP} (VAMP) was
proposed and analyzed in \cite{rangan2017vamp}, based closely on 
expectation propagation \cite{OppWin:05}---see also
\cite{fletcher2016expectation,ma2017orthogonal,takeuchi2017rigorous}.
There it was established that,
if $\Abf$ is a large right-rotationally invariant random matrix and
$\gbf(\cdot)$ is a separable Lipschitz denoiser, 
then VAMP's performance can be exactly predicted by a scalar SE, which also
provides testable conditions for optimality.
Importantly, VAMP applies to arbitrarily conditioned matrices $\Abf$,
which is a significant benefit over AMP,
since it is known that ill-conditioning is one of AMP's main failure
mechanisms~\cite{RanSchFle:14-ISIT,Caltagirone:14-ISIT,Vila:ICASSP:15}.


Unfortunately, the SE analyses of VAMP in \cite{rangan2016vamp} and its
extension in \cite{fletcher2017rigorous} are limited to separable denoisers.
This limitation prevents a full understanding of VAMP's behavior when used with
non-separable denoisers, such as state-of-the-art image-denoising methods
as recently suggested in \cite{schniter2017denoising}.
The main contribution of this work is to show that the SE analysis
of VAMP can be extended to a large class of non-separable denoisers
that are Lipschitz continuous and satisfy a certain convergence property.
The conditions are similar to those used in the analysis of
AMP with non-separable denoisers in \cite{berthier2017state}.
We show that there are several interesting non-separable denoisers that
satisfy these conditions, including group-structured and convolutional
neural network based denoisers.

For space considerations, all proofs and many details are provided in
Appendices in the Supplementary Materials section.

\section{Review of Vector AMP}

\begin{algorithm}[t]
\caption{Vector AMP (LMMSE form)}
\begin{algorithmic}[1]  \label{algo:vamp}
\REQUIRE{
LMMSE estimator
$\gbf_2(\cdot,\gamma_{2k})$ from \eqref{eq:g2slr},
denoiser $\gbf_1(\cdot,\gamma_{1k})$,
and
number of iterations $\Kit$.}%
\STATE{Select initial $\rbf_{10}$ and $\gamma_{10}\geq 0$.}
\FOR{$k=0,1,\dots,\Kit$}
    \STATE{// Denoising }
    \STATE{$\xbfhat_{1k} = \gbf_1(\rbf_{1k},\gamma_{1k})$}
        \label{line:x1}
    \STATE{$\alpha_{1k} = \bkt{ \nabla \gbf_1(\rbf_{1k},\gamma_{1k}) }$}
        \label{line:a1}
    \STATE{$\eta_{1k} = \gamma_{1k}/\alpha_{1k}$, $\gamma_{2k} = \eta_{1k} - \gamma_{1k}$}
        \label{line:gam2}
    \STATE{$\rbf_{2k} = (\eta_{1k}\xbfhat_{1k} - \gamma_{1k}\rbf_{1k})/\gamma_{2k}$}
        \label{line:r2}
    \STATE{ }
    \STATE{// LMMSE estimation }
    \STATE{$\xbfhat_{2k} = \gbf_2(\rbf_{2k},\gamma_{2k})$}
        \label{line:x2}
    \STATE{$\alpha_{2k} = \bkt{ \nabla \gbf_2(\rbf_{2k},\gamma_{2k}) } $}
        \label{line:a2}
    \STATE{$\eta_{2k} = \gamma_{2k}/\alpha_{2k}$, $\gamma_{1,\kp1} = \eta_{2k} - \gamma_{2k}$}
        \label{line:gam1}
    \STATE{$\rbf_{1,\kp1} = (\eta_{2k}\xbfhat_{2k} - \gamma_{2k}\rbf_{2k})/\gamma_{1,\kp1}$}
        \label{line:r1}
\ENDFOR
\STATE{Return $\xbfhat_{1\Kit}$.}
\end{algorithmic}
\end{algorithm}

The steps of VAMP algorithm of \cite{rangan2017vamp}
are shown in Algorithm~\ref{algo:vamp}.
Each iteration has two parts:  A denoiser step and a Linear MMSE (LMMSE)
step.  These are characterized by \emph{estimation functions} $\gbf_1(\cdot)$ and $\gbf_2(\cdot)$
producing estimates $\xbfhat_{1k}$ and $\xbfhat_{2k}$.  The estimation functions
take inputs $\rbf_{1k}$ and $\rbf_{2k}$ that we call \emph{partial estimates}.
The LMMSE estimation function is given by,
\begin{align}
\gbf_2(\rbf_{2k},\gamma_{2k})
&:= \left( \gamma_w \Abf\tran\Abf + \gamma_{2k}\Ibf\right)^{-1}
        \left( \gamma_w\Abf\tran\ybf + \gamma_{2k}\rbf_{2k} \right),
\label{eq:g2slr}
\end{align}
where $\gamma_w > 0$ is a parameter representing an estimate of the precision (inverse variance)
of the noise $\wbf$ in \eqref{eq:yAx}.  The estimate $\xbfhat_{2k}$ is thus an MMSE estimator,
treating the $\xbf$ as having a Gaussian prior with mean given by the partial estimate
$\rbf_{2k}$.  The estimation function $\gbf_1(\cdot)$ is called the \emph{denoiser} and
can be designed identically to the denoiser $\gbf(\cdot)$ in the AMP iterations
\eqref{eq:AMP}.
In particular, the denoiser is used to incorporate the structural or prior information
on $\xbf$.  As in AMP,
in lines~\ref{line:a1} and \ref{line:a2},  $\bkt{ \nabla \gbf_i }$ denotes the
normalized divergence.

The main result of \cite{rangan2016vamp} is that, under suitable conditions, VAMP
admits a state evolution (SE) analysis that precisely describes the mean squared error (MSE)
of the estimates $\xbfhat_{1k}$ and $\xbfhat_{2k}$ in a certain large system limit (LSL).
Importantly, VAMP's SE analysis applies to arbitrary right rotationally invariant $\Abf$.
This class is considerably larger than the set of sub-Gaussian i.i.d.\ matrices
for which AMP applies. However, the SE analysis
in \cite{rangan2016vamp} is restricted
separable Lipschitz denoisers that can be described as follows:
Let $g_{1n}(\rbf_1,\gamma_1)$ be
the $n$-th component of the output of $\gbf_1(\rbf_1,\gamma_1)$.  Then, it is assumed that,
\beq \label{eq:g1sep}
    \xhat_{1n} = g_{1n}(\rbf_1,\gamma_1) = \phi(r_{1n},\gamma_1),
\eeq
for some function scalar-output function $\phi(\cdot)$ that does not depend on the component index $n$.
Thus, the estimator is separable in the sense that the
$n$-th component of the estimate, $\xhat_{1n}$ depends only on the $n$-th component of the input $r_{1n}$ as well as the precision level $\gamma_1$.  In addition, it is assumed that
$\phi(r_1,\gamma_1)$ satisfies a certain Lipschitz condition.  The separability assumption
precludes the analysis of more general denoisers mentioned in the Introduction.

%
%

\section{Extending the Analysis to Non-Separable Denoisers}
\label{sec:denoiser_ex}

The main contribution of the paper is to extend the state evolution analysis of VAMP
to a class of denoisers
that we call \emph{uniformly Lipschitz} and
\emph{convergent under Gaussian noise}.  This class is significantly
larger than separable Lipschitz denoisers used in \cite{rangan2016vamp}.
To state these conditions precisely,
consider a sequence of estimation problems, indexed by a vector dimension $N$.
For each $N$, suppose there is some ``true" vector $\ubf = \ubf(N) \in \R^N$ that we wish to estimate
from noisy measurements of the form, $\rbf = \ubf + \zbf$, where $\zbf \in \R^N$ is
Gaussian noise.
Let $\ubfhat = \gbf(\rbf,\gamma)$ be some estimator, parameterized by $\gamma$.

\begin{definition}  \label{def:uniflip}
The sequence of estimators $\gbf(\cdot)$ are said to
be \underline{uniformly Lipschitz continuous} if
there exists constants $A$, $B$ and $C > 0$, such that
\beq \label{eq:lipcond}
    \|\gbf(\rbf_2,\gamma_2)-\gbf(\rbf_1,\gamma_1)\| \leq
    (A + B|\gamma_2-\gamma_1|)\|\rbf_2-\rbf_1\| + C\sqrt{N}|\gamma_2-\gamma_1|,
\eeq
for any $\rbf_1,\rbf_2,\gamma_1,\gamma_2$ and $N$.
\end{definition}

\begin{definition}  \label{def:conv}
The sequence of random vectors $\ubf$ and estimators $\gbf(\cdot)$ are said to
be \underline{convergent under Gaussian noise} if the following condition holds:
Let $\zbf_1,\zbf_2 \in \R^N$ be two sequences
where $(z_{1n},z_{2n})$ are i.i.d.\ with  $(z_{1n},z_{2n})=\Norm(0,\Sbf)$
for some positive definite covariance $\Sbf \in \R^{2 \x 2}$.
Then, all the following limits exist almost surely:
\begin{subequations}\label{eq:convlim}
\begin{align}
    &\lim_{N \arr \infty} \frac{1}{N} \gbf(\ubf+\zbf_1,\gamma_1)\tran \gbf(\ubf+\zbf_2,\gamma_2),
    \quad
    \lim_{N \arr \infty} \frac{1}{N} \gbf(\ubf+\zbf_1,\gamma_1)\tran \ubf,
    \label{eq:momlim} \\
   & \lim_{N \arr \infty} \frac{1}{N} \ubf\tran\zbf_1, \quad
    \lim_{N \arr \infty} \frac{1}{N} \|\ubf\|^2   \\
    &\lim_{N \arr \infty} \bkt{\nabla \gbf(\ubf+\zbf_1,\gamma_1)}
        = \frac{1}{NS_{12}} \gbf(\ubf+\zbf_1,\gamma_1)\tran \zbf_2,
        \label{eq:divlim}
\end{align}
\end{subequations}
for all $\gamma_1,\gamma_2$ and covariance matrices $\Sbf$. 
Moreover, the values of the limits are continuous in $\Sbf$, $\gamma_1$ and $\gamma_2$.
\end{definition}

\medskip
With these definitions, we make the following key assumption on the denoiser.

\begin{assumption} \label{as:denoiser}
For each $N$, suppose that we have a ``true" random vector $\xbf^0 \in \R^N$
and a denoiser $\gbf_1(\rbf_1,\gamma_1)$ acting on signals $\rbf_1 \in \R^N$.
Following Definition~\ref{def:uniflip}, we assume
the sequence of denoiser functions
indexed by $N$, is uniformly Lipschitz continuous.
In addition, the sequence of true
vectors $\xbf^0$ and denoiser functions are convergent under Gaussian noise
following Definition~\ref{def:conv}.
\end{assumption}

\medskip
The first part of Assumption~\ref{as:denoiser} is relatively standard:
Lipschitz and uniform Lipschitz continuity of the denoiser is assumed several AMP-type
analyses including \cite{BayatiM:11,KamRanFU:12-IT,rangan2016vamp}
What is new is the assumption in Definition~\ref{def:conv}.
This assumption
relates to the behavior of the denoiser $\gbf_1(\rbf_1,\gamma_1)$ in the case
when the input is of the form, $\rbf_1 = \xbf^0 + \zbf$.  That is,
the input is the true signal with a Gaussian noise perturbation.
In this setting, we will be requiring that certain correlations converge.
Before continuing our analysis, we briefly show that separable denoisers as well
as several interesting
non-separable denoisers satisfy these conditions.

\paragraph*{Separable Denoisers.}
We first show that the class of denoisers satisfying Assumption~\ref{as:denoiser}
includes the separable
Lipschitz denoisers studied in most AMP analyses such as \cite{BayatiM:11}.
Specifically, suppose that the true vector $\xbf^0$ has i.i.d.\ components
with bounded second moments and the denoiser
$\gbf_1(\cdot)$ is separable in that it is of the form \eqref{eq:g1sep}.
Under a certain uniform Lipschitz condition, it is shown in
Appendix~\ref{sec:denoiser_ex_det} that the denoiser satisfies Assumption~\ref{as:denoiser}.

\paragraph*{Group-Based Denoisers.}
As a first non-separable example, let us suppose that the vector
$\xbf^0$ can be represented as an $L \x K$ matrix.
Let $\xbf^0_\ell \in \R^K$ denote the $\ell$-th row and assume that the rows are i.i.d.
Each row can represent a \emph{group}.
Suppose that the denoiser $\gbf_1(\cdot)$ is \emph{groupwise separable}.
That is, if we denote by $\gbf_{1\ell}(\rbf,\ell)$ the
$\ell$-th row of the output of the denoiser, we assume that
\beq \label{eq:ggroup}
    \gbf_{1\ell}(\rbf,\gamma) = \mathbf{\phi}(\rbf_\ell,\gamma) \in \R^K,
\eeq
for a vector-valued function $\mathbf{\phi}(\cdot)$ that is the same for all rows.
Thus, the $\ell$-th row output $\gbf_\ell(\cdot)$ depends only on the $\ell$-th row input.
Such groupwise denoisers have been used in AMP and EP-type methods for group LASSO
and other structured estimation problems~\cite{taeb2013maximin,andersen2014bayesian,rangan2017hybrid}.
Now, consider the limit where the group size $K$ is fixed, and the number of groups $L \arr \infty$.
Then, under suitable Lipschitz continuity conditions,
Appendix~\ref{sec:denoiser_ex_det} shows that groupwise separable denoiser
also satisfies Assumption~\ref{as:denoiser}.

\paragraph*{Convolutional Denoisers.}
As another non-separable denoiser,
suppose that, for each $N$, $\xbf^0$ is an $N$ sample segment of a stationary,
ergodic process with bounded second moments.
Suppose that the denoiser is given by a linear convolution,
\beq \label{eq:gconvolution}
    \gbf_1(\rbf_1) := T_N(\hbf * \rbf_1),
\eeq
where $\hbf$ is a finite length filter and $T_N(\cdot)$ truncates the signal to its
first $N$ samples.  For simplicity,
we assume there is no dependence on $\gamma_1$.
Convolutional denoising arises in many standard linear estimation operations
on wide sense stationary
processes such as Weiner filtering and smoothing~\cite{scharf1991statistical}.
If we assume that $\hbf$ remains constant and $N \arr \infty$,
Appendix~\ref{sec:denoiser_ex_det} shows that the sequence of random vectors
$\xbf^0$ and convolutional denoisers $\gbf_1(\cdot)$ satisfies Assumption~\ref{as:denoiser}.

\paragraph*{Convolutional Neural Networks.}  In recent years, there has been considerable interest
in using trained deep convolutional neural networks for image denoising
\cite{xie2012image,xu2014deep}.  As a simple model for such a denoiser,
suppose that the denoiser is a composition of maps,
\beq \label{eq:gcnn}
    \gbf_1(\rbf_1) = (F_L \circ F_{L-1} \circ \cdots \circ F_1)(\rbf_1),
\eeq
where  $F_\ell(\cdot)$ is a sequence of layer maps where each layer is either
a multi-channel convolutional operator or Lipschitz separable activation function,
such as sigmoid or ReLU.  Under mild assumptions on the maps,
it is shown in Appendix~\ref{sec:denoiser_ex_det}
the estimator sequence $\gbf_1(\cdot)$ can also satisfy Assumption~\ref{as:denoiser}.

\paragraph*{Singular-Value Thresholding (SVT) Denoiser.} 
Consider the estimation of a low-rank matrix $\Xbf^0$ from linear measurements $\ybf=\mathcal{A}(\Xbf^0)$, where $\mathcal{A}$ is some linear operator~\cite{cai2010singular}. 
Writing the SVD of $\Rbf$ as $\Rbf=\sum_i\sigma_i\ubf_i\vbf_i\tran$, the SVT denoiser is defined as
\begin{align}\label{eq:svt}
    \gbf_1(\Rbf,\gamma)\defn\sum_i(\sigma_i-\gamma)_+\ubf_i\vbf_i\tran,
\end{align}
where $(x)_+\defn\text{max}\{0,x\}$. 
In Appendix~\ref{sec:denoiser_ex_det}, we show that $\gbf_1(\cdot)$ satisfies Assumption~\ref{as:denoiser}.

\section{Large System Limit Analysis}

\subsection{System Model}

Our main theoretical contribution is to show that the SE analysis of VAMP in
\cite{rangan2017vamp} can be extended to the non-separable case.
We consider a sequence of problems indexed by the vector dimension $N$.
For each $N$, we assume that there is a ``true" random
vector $\xbf^0\in\R^N$ observed through measurements $\ybf\in\R^M$ of the form in \eqref{eq:yAx}
where $\wbf \sim \Norm(\mathbf{0},\gamma_{w0}^{-1}\Ibf)$.
We use $\gamma_{w0}$ to denote the ``true" noise precision to distinguish this from the postulated
precision, $\gamma_w$, used in the LMMSE estimator \eqref{eq:g2slr}.
Without loss of generality (see below), we assume that $M=N$.
We assume that $\Abf$ has an SVD,
\beq \label{eq:ASVD}
    \Abf=\Ubf\Sbf\Vbf\tran, \quad \Sbf=\mathrm{diag}(\sbf), \quad \sbf = (s_1,\ldots,s_N),
\eeq
where $\Ubf$ and $\Vbf$ are orthogonal and $\Sbf$ is non-negative and diagonal.
The matrix $\Ubf$ is arbitrary, $\sbf$ is an i.i.d.\ random vector with components
$s_i \in [0,s_{max}]$ almost surely.  Importantly, we assume that $\Vbf$ is Haar distributed,
meaning that it is uniform on the $N \x N$ orthogonal matrices.
This implies that $\Abf$ is \emph{right rotationally invariant} meaning that
$\Abf \eqd \Abf\Vbf_0$ for any orthogonal matrix $\Vbf_0$.   We also assume that
$\wbf$, $\xbf^0$, $\sbf$ and $\Vbf$ are all independent.  As in \cite{rangan2017vamp},
we can handle the case of rectangular $\Vbf$ by zero padding $\sbf$.

These assumptions are similar to those in \cite{rangan2017vamp}.
The key new assumption is Assumption~\ref{as:denoiser}.
Given such a denoiser and postulated variance $\gamma_w$, we run the VAMP algorithm,
Algorithm~\ref{algo:vamp}.  We assume that the initial condition is given by,
\beq \label{eq:r1init}
    \rbf = \xbf^0 + \Norm(\mathbf{0},\tau_{10}\Ibf),
\eeq
for some initial error variance $\tau_{10}$.  In addition, we assume
\beq \label{eq:gam10lim}
    \lim_{N \arr \infty} \gamma_{10} = \gammabar_{10},
\eeq
almost surely for some $\gammabar_{10} \geq 0$.

Analogous to \cite{rangan2016vamp}, we define two key functions: \emph{error functions}
and \emph{sensitivity functions}.
The error functions characterize the MSEs of the denoiser and LMMSE estimator under AWGN measurements.
For the denoiser $\gbf_1(\cdot,\gamma_1)$, we define the error function as
\begin{align}
    \MoveEqLeft \Ecal_1(\gamma_1,\tau_1)
    := \lim_{N \arr \infty} \frac{1}{N} \|\gbf_1(\xbf^0+\zbf,\gamma_1)-\xbf^0\|^2, \quad
    \zbf \sim \Norm(\mathbf{0}, \tau_1\Ibf),
\label{eq:eps1}
\end{align}
and, for the LMMSE estimator, as
\begin{align}
    \MoveEqLeft \Ecal_2(\gamma_2,\tau_2)  := \lim_{N \arr \infty}
        \frac{1}{N} \Exp \| \gbf_2(\rbf_2,\gamma_2) -\xbf^0 \|^2, \nonumber \\
    & \rbf_2 = \xbf^0 + \Norm(0,\tau_2 \Ibf), \quad
    \ybf = \Abf\xbf^0 + \Norm(0,\gamma_{w0}^{-1} \Ibf).
    \label{eq:eps2}
\end{align}
The limit \eqref{eq:eps1} exists almost surely due to the assumption of $\gbf_1(\cdot)$ being
convergent under Gaussian noise.
Although $\Ecal_2(\gamma_2,\tau_2)$ implicitly depends on the precisions $\gamma_{w0}$ and $\gamma_w$, we omit this dependence to simplify the notation.
We also define the \emph{sensitivity functions} as
\beq \label{eq:Ai}
    {\mathcal A}_i(\gamma_i,\tau_i) := \lim_{N \arr \infty}
        \bkt{ \nabla \gbf_i(\xbf^0 + \zbf_i,\gamma_i) }, \quad \zbf_i \sim \Norm(\mathbf{0},\tau_i\Ibf).
\eeq
The LMMSE error function \eqref{eq:eps2} and sensitivity functions \eqref{eq:Ai} are
identical to those in the VAMP analysis \cite{rangan2017vamp}.  The denoiser
error function \eqref{eq:eps1} generalizes the error function in \cite{rangan2017vamp} for non-separable
denoisers.

\subsection{State Evolution of VAMP}

We now show that the VAMP algorithm with a non-separable denoiser
follows the identical state evolution equations as the separable case given in \cite{rangan2017vamp}.
Define the error vectors,
\beq \label{eq:pqdef}
    \pbf_k := \rbf_{1k}-\xbf^0, \quad \qbf_k := \Vbf\tran(\rbf_{2k}-\xbf^0).
\eeq
Thus, $\pbf_k$ represents the error between the partial estimate $\rbf_{1k}$ and
the true vector $\xbf^0$.  The error vector $\qbf_k$ represents the
transformed error $\rbf_{2k}-\xbf^0$.  The SE analysis will show that these errors
are asymptotically Gaussian.  In addition, the analysis will exactly predict the
variance on the partial estimate errors \eqref{eq:pqdef} and
estimate errors, $\xbfhat_i-\xbf^0$.  These variances are computed recursively
through what we will call the \emph{state evolution} equations:
\begin{subequations} \label{eq:se}
\begin{align}
    \alphabar_{1k} &= {\mathcal A}_1(\gammabar_{1k},\tau_{1k}), \quad
    \etabar_{1k} = \frac{\gammabar_{1k}}{\alphabar_{1k}}, \quad
    \gammabar_{2k} = \etabar_{1k} - \gammabar_{1k} \label{eq:a1se} \\
    \tau_{2k} &= \frac{1}{(1-\alphabar_{1k})^2}\left[
        \Ecal_1(\gammabar_{1k},\tau_{1k}) - \alphabar_{1k}^2\tau_{1k} \right],
            \label{eq:tau2se} \\
    \alphabar_{2k} &= {\mathcal A}_2(\gammabar_{2k},\tau_{2k}), \quad
    \etabar_{2k} = \frac{\gammabar_{2k}}{\alphabar_{2k}}, \quad
    \gammabar_{1,\kp1} = \etabar_{2k} - \gammabar_{2k} \label{eq:a2se} \\
    \tau_{1,\kp1} &= \frac{1}{(1-\alphabar_{2k})^2}\left[
        \Ecal_2(\gammabar_{2k},\tau_{2k}) - \alphabar_{2k}^2\tau_{2k} \right],
            \label{eq:tau1se}
\end{align}
\end{subequations}
which are initialized with $k=0$, $\tau_{10}$ in \eqref{eq:r1init}
and  $\gammabar_{10}$ defined from the limit \eqref{eq:gam10lim}.
The SE equations in \eqref{eq:se} are identical to those in \cite{rangan2017vamp}
with the new error and sensitivity functions for the non-separable denoisers.
We can now state our main result.

\begin{theorem} \label{thm:se}
Under the above assumptions and definitions, assume
that the sequence of true random vectors $\xbf^0$ and denoisers
$\gbf_1(\rbf_1,\gamma_1)$ satisfy Assumption~\ref{as:denoiser}.
Assume additionally that, for all iterations $k$,
the solution $\alphabar_{1k}$ from the SE equations \eqref{eq:se} satisfies
$\alphabar_{1k} \in (0,1)$ and $\gammabar_{ik} > 0$.
Then,
\begin{enumerate}[(a)]
\item For any $k$, the error vectors on the partial estimates,
$\pbf_k$ and $\qbf_k$ in \eqref{eq:pqdef} can be written as,
\beq \label{eq:pqgauss}
    \pbf_k = \tilde{\pbf}_k + \ONinv, \quad \qbf_k = \tilde{\qbf}_k + \ONinv,
\eeq
where, $\tilde{\pbf}_k$ and $\tilde{\qbf}_k \in \R^N$ are each
 i.i.d.\ Gaussian random vectors with zero
mean and per component variance $\tau_{1k}$ and $\tau_{2k}$, respectively.

\item For any fixed iteration $k \geq 0$, and $i=1,2$, we have, almost surely 
\beq \label{eq:aplim}
    \lim_{N \arr \infty} \frac{1}{N} \|\xbfhat_i-\xbf^0\|^2 = \frac{1}{\etabar_{ik}}, \quad
    \lim_{N \arr \infty} (\alpha_{ik},\eta_{ik},\gamma_{ik}) =
    (\alphabar_{ik},\etabar_{ik}, \gammabar_{ik}).
\eeq

\end{enumerate}
\end{theorem}
\begin{proof}
See Appendix \ref{sec:sepf}.
\end{proof}

In \eqref{eq:pqgauss}, we have used the notation, that when
$\ubf,\tilde{\ubf} \in \R^N$ are sequences of random vectors,
$\ubf = \tilde{\ubf} + \ONinv$ means
$\lim_{N \arr \infty} \tfrac{1}{N} \|\ubf -\tilde{\ubf}\|^2 = 0$
almost surely.
Part (a) of Theorem~\ref{thm:se} thus shows that the error vectors
$\pbf_k$ and $\qbf_k$ in \eqref{eq:pqdef} are approximately i.i.d.\ Gaussian.
The result is a natural extension to the main result on separable denoisers
in \cite{rangan2017vamp}. Moreover, the variance on the variance on the errors,
along with the mean squared error (MSE) of the estimates $\xbfhat_{ik}$ can
be exactly predicted by the same SE equations as the separable case.
The result thus provides an asymptotically exact analysis of VAMP extended
to non-separable denoisers.

\section{Numerical Experiments} \label{sec:num}

\subsection{Compressive Image Recovery}

We first consider the problem of compressive image recovery, 
where the goal is to recover an image $\xbf^0\in\R^N$ from 
measurements $\ybf\in\R^M$ of the form \eqref{eq:yAx} with $M\ll N$.
This problem arises in many imaging applications, such as
magnetic resonance imaging, radar imaging, computed tomography, etc.,
although the details of $\Abf$ and $\xbf^0$ change in each case.

\begin{figure}[t]
\newcommand{\sz}{0.6}
\centering
\begin{subfigure}[t]{0.47\textwidth}
  \centering
  \psfrag{sampling rate M/N}[t][t][\sz]{\sf sampling ratio $M/N$}
  \psfrag{PSNR}[t][t][\sz]{\sf PSNR}
  \psfrag{runtime (sec)}[t][t][\sz]{\sf runtime (sec)}
  \includegraphics[width=\columnwidth,
                   clip=true]{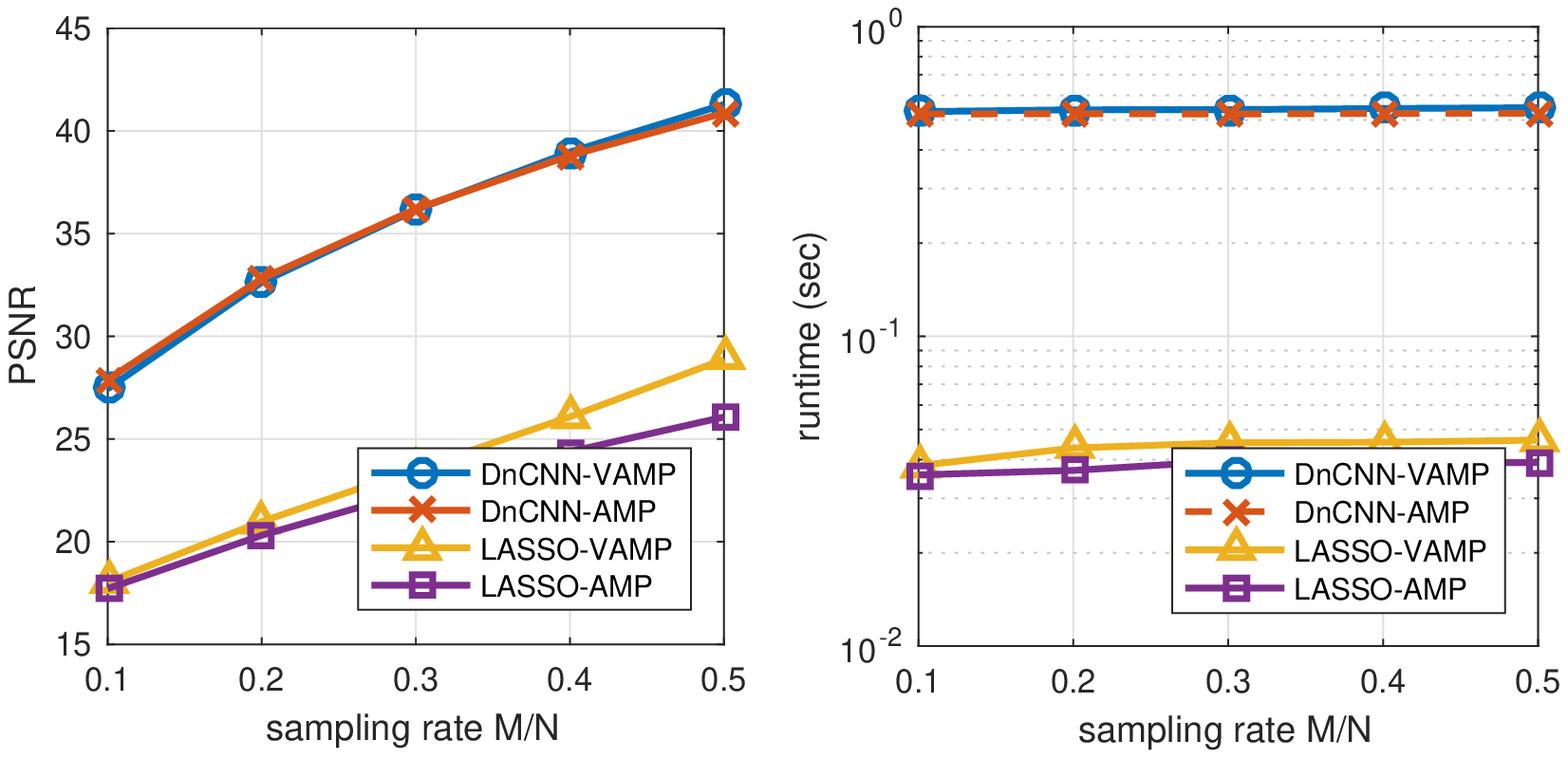}
  \caption{Average PSNR and runtime with vs.\ $M/N$ 
           with well-conditioned $\Abf$ and no noise after 12 iterations.}
  \label{fig:vsRate}
\end{subfigure}
\hfill
\begin{subfigure}[t]{0.47\textwidth}
  \centering
  \psfrag{condition number}[t][t][\sz]{\sf $\text{\sf cond}(\Abf)$}
  \psfrag{PSNR}[t][t][\sz]{\sf PSNR}
  \psfrag{runtime (sec)}[t][t][\sz]{\sf runtime (sec)}
  \includegraphics[width=0.98\columnwidth,
                   clip=true]{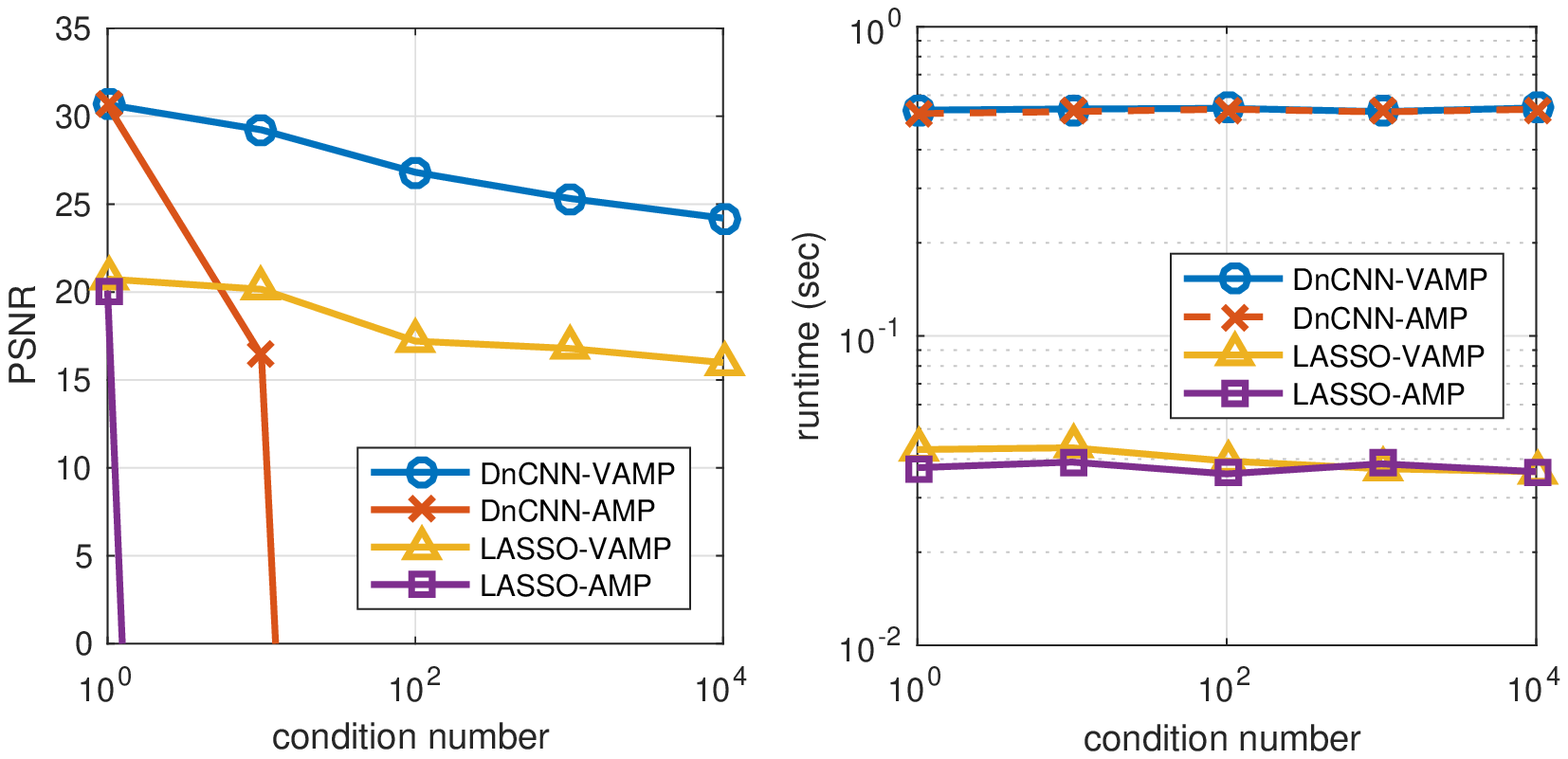}
  \caption{Average PSNR and runtime versus $\text{cond}(\Abf)$
           at $M/N=0.2$ and no noise after 10 iterations.}
  \label{fig:vsCondNum}
\end{subfigure}
\caption{Compressive image recovery: PSNR and runtime vs.\ rate $M/N$ and $\text{cond}(\Abf)$}
\end{figure}

One of the most popular approaches to image recovery is to exploit sparsity 
in the wavelet transform coefficients $\cbf\defn\Psibf\xbf^0$, 
where $\Psibf$ is a suitable orthonormal wavelet transform.
Rewriting \eqref{eq:yAx} as $\ybf=\Abf\Psibf\cbf + \wbf$, the idea is to first
estimate $\cbf$ from $\ybf$ (e.g., using LASSO) and then form the 
image estimate via $\xbfhat=\Psibf\tran\cbfhat$.
Although many algorithms exist to solve the LASSO problem, the AMP algorithms
are among the fastest (see, e.g., \cite[Fig.1]{borgerding2017amp}).
As an alternative to the sparsity-based approach, it was recently suggested
in \cite{Metzler2016denoise}
to recover $\xbf^0$ directly using AMP \eqref{eq:AMP} by choosing the 
estimation function $\gbf$ as a sophisticated image-denoising algorithm like 
BM3D \cite{dabov2007image} or
DnCNN \cite{zhang2017beyond}.

Figure~\ref{fig:vsRate} compares the LASSO- and DnCNN-based versions
of AMP and VAMP for 128$\times$128 image recovery under 
well-conditioned $\Abf$
and no noise.
Here, 
$\Abf=\Jbf\Pbf\Hbf\Dbf$, where
$\Dbf$ is a diagonal matrix with random $\pm 1$ entries,
$\Hbf$ is a discrete Hadamard transform (DHT),
$\Pbf$ is a random permutation matrix, and
$\Jbf$ contains the first $M$ rows of $\Ibf_N$. 
The results average over the well-known
\textsl{lena}, \textsl{barbara}, \textsl{boat},
\textsl{house}, and \textsl{peppers} images 
using 10 random draws of $\Abf$ for each.
The figure shows that AMP and VAMP have very similar runtimes and PSNRs 
when $\Abf$ is 
well-conditioned,
and that
the DnCNN approach is about 10~dB more accurate, but 10$\times$ as slow, as the LASSO approach. 
Figure~\ref{fig:stateEvo} shows the state-evolution prediction of VAMP's PSNR on the \emph{barbara} image at $M/N=0.5$,
averaged over 50 draws of $\Abf$. 
The state-evolution accurately predicts the PSNR of VAMP.

To test the robustness to the condition number of $\Abf$, 
we repeated the experiment from Fig.~\ref{fig:vsRate} using 
$\Abf=\Jbf\Diag(\sbf)\Pbf\Hbf\Dbf$, 
where 
$\Diag(\sbf)$ is a diagonal matrix of singular values.
The singular values were geometrically spaced,
i.e., $s_m/s_{m-1}=\rho~\forall m$, with $\rho$ chosen to achieve a desired 
$\text{cond}(\Abf)\defn s_1/s_M$.
The sampling rate was fixed at $M/N=0.2$, and
the measurements were noiseless, as before.
The results, shown in Fig.~\ref{fig:vsCondNum}, show that
AMP diverged when $\text{cond}(\Abf)\geq 10$, while 
VAMP exhibited only a mild PSNR degradation due to ill-conditioned $\Abf$.
The original images and example image recoveries are included in 
Appendix~\ref{sec:image} of the supplementary material.

\subsection{Bilinear Estimation via Lifting}

We now use the structured linear estimation model \eqref{eq:yAx} 
to tackle problems in \emph{bilinear} estimation through 
a technique known as ``lifting'' 
\cite{candes2013phaselift,ahmed2014blind,ling2015self,davenport2016overview}.
In doing so, we are motivated by applications like
blind deconvolution \cite{Haykin:94},
self-calibration \cite{ling2015self},
compressed sensing (CS) with matrix uncertainty \cite{zhu2011sparsity},
and joint channel-symbol estimation \cite{sun2018joint}.
All cases yield measurements $\ybf$ of the form
\begin{equation}
\label{eq:bilinear}
\ybf = \textstyle \big(\sum_{l=1}^L b_l \Phibf_l \big) \cbf + \wbf \in \R^M,
\end{equation}
where 
$\{\Phibf_l\}_{l=1}^L$ are known, 
$\wbf\sim\Norm(\zero,\Ibf/\gamma_w)$, 
and the objective is to recover both 
$\bbf\defn[b_1,\dots,b_L]\tran$ and $\cbf\in\R^P$.
This bilinear problem can be ``lifted'' into a linear problem 
of the form \eqref{eq:yAx} by setting
\begin{equation}
\label{eq:lifting}
\Abf = \left[\begin{matrix}\Phibf_1&
                           \Phibf_2&\cdots&
                           \Phibf_L\end{matrix}
       \right]\in\R^{M\times LP}
\text{~and~}
\xbf = \text{vec}(\cbf\bbf\tran) \in \R^{LP},
\end{equation}
where $\text{vec}(\Xbf)$ vectorizes $\Xbf$ by concatenating its columns.
When $\bbf$ and $\cbf$ are i.i.d.\ with known~priors, the MMSE denoiser 
$\gbf(\rbf,\gamma) =\Exp(\xbf|\rbf=\xbf+\Norm(\zero,\Ibf/\gamma))$
can be implemented near-optimally by the rank-one AMP algorithm 
from \cite{RanganF:12-ISIT} (see also 
\cite{deshpande2014information,matsushita2013lowrank,lesieur2015phase}),
with divergence estimated as in \cite{Metzler2016denoise}. \quad 

We first consider \emph{CS with matrix uncertainty}
\cite{zhu2011sparsity}, where $b_1$ is known. 
For these experiments, we generated 
the unknown $\{b_l\}_{l=2}^L$ as i.i.d.\ $\Norm(0,1)$ and
the unknown $\cbf\in\R^P$ as $K$-sparse with $\Norm(0,1)$ nonzero entries.
Fig.~\ref{fig:stateEvo} shows that the MSE on $\xbf$ of lifted VAMP is very close to its SE prediction when $K=12$. 
We then compared lifted VAMP to PBiGAMP from \cite{parker2016bilinear}, 
which applies AMP directly to the (non-lifted) bilinear problem, and
to WSS-TLS from \cite{zhu2011sparsity}, 
which uses non-convex optimization.
We also compared to MMSE estimation of $\bbf$ under oracle knowledge of $\cbf$, 
and MMSE estimation of $\cbf$ under 
oracle knowledge of $\text{support}(\cbf)$ and $\bbf$.
For $b_1=\sqrt{20}$, $L=11$, $P=256$, $K=10$, 
i.i.d.\ $\Norm(0,1)$ matrix $\Abf$, and SNR $=$ 40~dB, 
Fig.~\ref{fig:csmu_delta} shows the normalized MSE on $\bbf$ 
(i.e., $\text{\sf NMSE}(\bbf)\defn\Exp\|\bbfhat-\bbf^0\|^2/\Exp\|\bbf^0\|^2$) 
and $\cbf$ versus sampling ratio $M/P$.
This figure demonstrates that lifted VAMP and PBiGAMP perform close to the oracles and much better than WSS-TLS.

Although lifted VAMP performs similarly to PBiGAMP in Fig.~\ref{fig:csmu_delta},
its advantage over PBiGAMP becomes apparent with non-i.i.d.\ $\Abf$.
For illustration, we repeated the previous experiment, 
but with $\Abf$ constructed using the SVD $\Abf=\Ubf\Diag(\sbf)\Vbf\tran$ 
with Haar distributed $\Ubf$ and $\Vbf$ and geometrically spaced $\sbf$.
Also, to make the problem more difficult, we set $b_1=1$.
Figure~\ref{fig:csmu_condA} shows the normalized MSE on $\bbf$ and $\cbf$ 
versus $\text{cond}(\Abf)$ at $M/P=0.6$.
There it can be seen that lifted VAMP is much more robust than PBiGAMP 
to the conditioning of $\Abf$.

\begin{figure}[t]
\begin{minipage}{0.48\textwidth}
  \newcommand{\sz}{0.6}
  \psfrag{iterations}[t][t][\sz]{\sf iteration}
  \psfrag{PSNR (dB)}[b][b][\sz]{\sf PSNR in dB}
  \psfrag{NMSE}[b][b][\sz]{\sf NMSE in dB}
  \psfrag{imaging}[b][b][\sz]{\sf image recovery}
  \psfrag{CSMU}[b][b][\sz]{\sf CS with matrix uncertainty}
  \includegraphics[scale=0.205,
                  trim=8mm 0mm 15mm 3mm, 
                  clip=true]{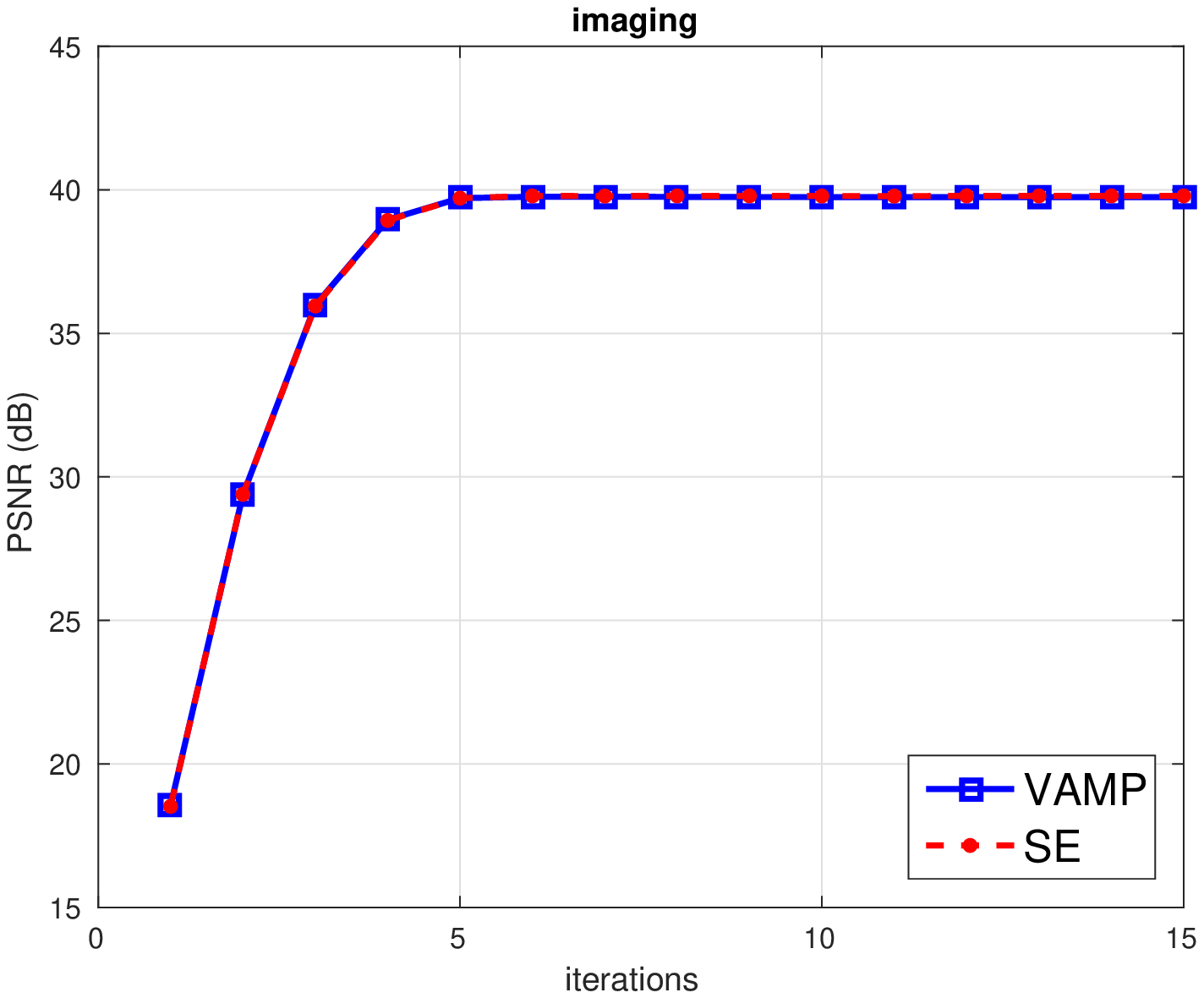}
  \hspace{3mm}
  \includegraphics[scale=0.195,
                  trim=8mm 0mm 15mm 3mm, 
                  clip=true]{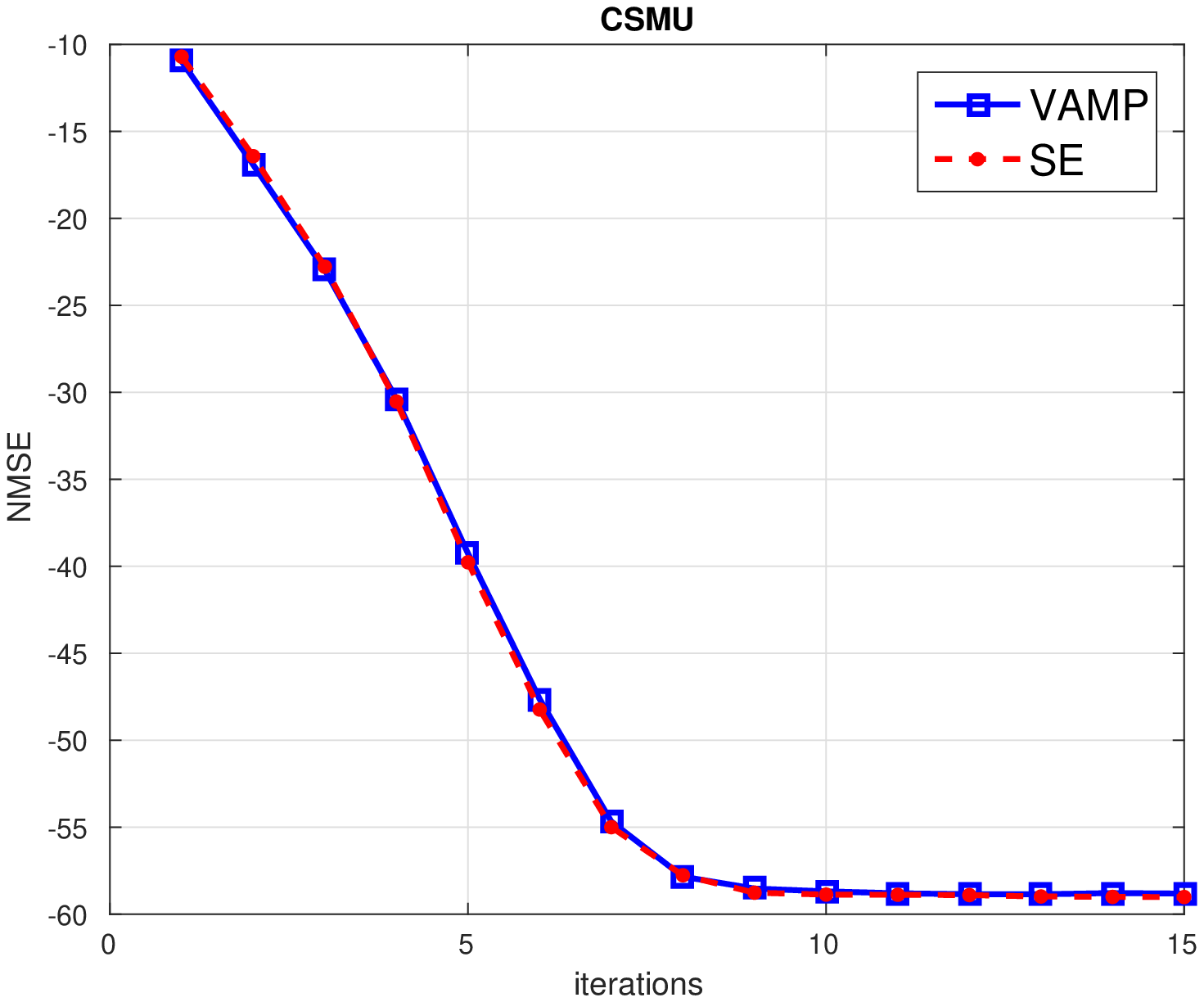}
  \hspace{3mm}
  \caption{SE prediction \& VAMP for image recovery and CS with matrix uncertainty}
  \label{fig:stateEvo}
\end{minipage}
\hfill
\begin{minipage}{0.48\textwidth}
  \centering
  \newcommand{\sz}{0.6}
  \psfrag{Nb}[c][b][\sz]{\sf subspace dimension $L$}
  \psfrag{K}[t][t][\sz]{\sf sparsity $K$}
  \psfrag{VAMP-Lift}[b][b][\sz]{\sf Lifted VAMP}
  \psfrag{SparseLift}[b][b][\sz]{\sf SparseLift}
  \includegraphics[scale=0.38,
                   trim=19mm 19mm 22mm 20mm, 
                   clip=true]{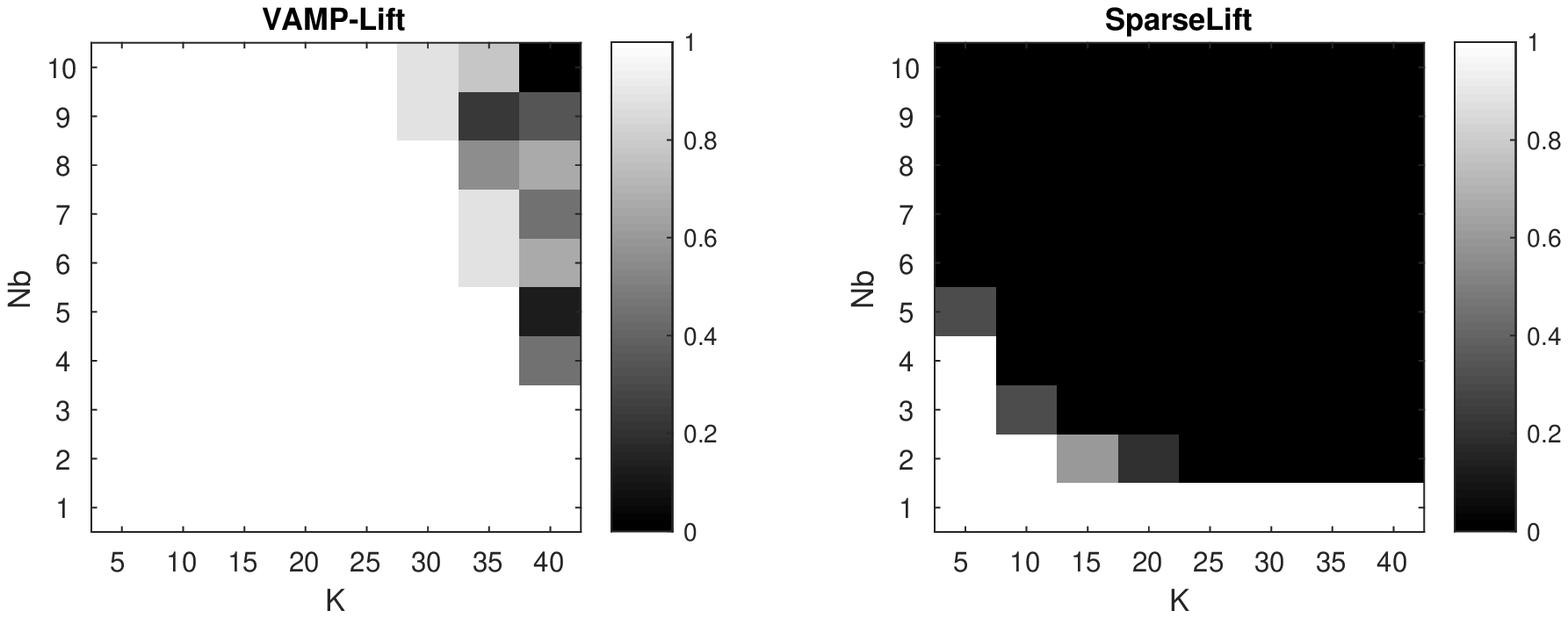}
  \caption{Self-calibration: Success rate vs.\ sparsity $K$ and subspace dimension $L$}
  \label{fig:calibration}
\end{minipage}
\end{figure}

\begin{figure}[t]
\newcommand{\sz}{0.6}
\centering
\begin{subfigure}[t]{0.47\textwidth}
  \centering
  \psfrag{sampling ratio M/N}[t][t][\sz]{\sf sampling ratio $M/P$}
  \psfrag{NMSE(b) [dB]}[b][b][\sz]{\sf NMSE($\bbf$) in dB}
  \psfrag{NMSE(c) [dB]}[b][b][\sz]{\sf NMSE($\cbf$) in dB}
  \includegraphics[width=\columnwidth,
                   trim=5mm 5mm 15mm 10mm, 
                   clip=true]{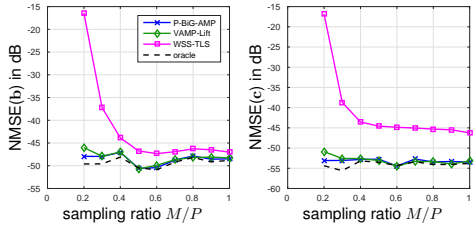}
  \caption{NMSE vs.\ $M/P$ with i.i.d.~$\Norm(0,1)$ $\Abf$.}
  \label{fig:csmu_delta}
\end{subfigure}
\hspace{3mm}
\begin{subfigure}[t]{0.47\textwidth}
  \centering
  \psfrag{cond(A)}[t][t][\sz]{\sf $\text{\sf cond}(\Abf)$}
  \psfrag{NMSE(b) [dB]}[b][b][\sz]{\sf NMSE($\bbf$) in dB}
  \psfrag{NMSE(c) [dB]}[b][b][\sz]{\sf NMSE($\cbf$) in dB}
  \includegraphics[width=\columnwidth,
                   trim=5mm 5mm 15mm 10mm, 
                   clip=true]{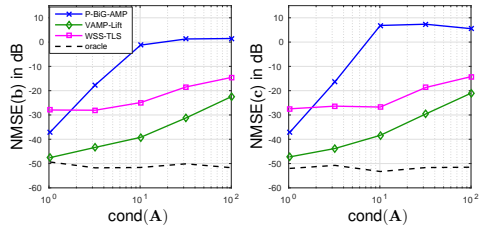}
  \caption{NMSE vs.\ $\text{cond}(\Abf)$ at $M/P=0.6$.}
  \label{fig:csmu_condA}
\end{subfigure}
\caption{Compressive sensing with matrix uncertainty}
\end{figure}

We next consider the \emph{self-calibration} problem \cite{ling2015self}, 
where the measurements take the form
\begin{equation}
\label{eq:selfcal}
\ybf = \textstyle \Diag(\Hbf\bbf) \Psibf \cbf + \wbf \in \R^M .
\end{equation}
Here the matrices $\Hbf\in\R^{M\times L}$ and $\Psibf\in\R^{M\times P}$ are 
known and the objective is to recover the unknown vectors $\bbf$ and $\cbf$.
Physically, the vector $\Hbf\bbf$ represents unknown calibration gains 
that lie in a known subspace, specified by $\Hbf$.
Note that \eqref{eq:selfcal} is an instance of \eqref{eq:bilinear} with $\Phibf_l=\Diag(\hbf_l)\Psibf$,
where $\hbf_l$ denotes the $l$th column of $\Hbf$.
Different from ``CS with matrix uncertainty,'' all elements in $\bbf$ 
are now unknown, and so WSS-TLS \cite{zhu2011sparsity} cannot be applied.
Instead, we compare lifted VAMP to the SparseLift approach from 
\cite{ling2015self}, which is based on convex relaxation and has provable
guarantees.
For our experiment, we generated 
$\Psibf$ and $\bbf\in\R^L$ as i.i.d.\ $\Norm(0,1)$; 
$\cbf$ as $K$-sparse with $\Norm(0,1)$ nonzero entries; 
$\Hbf$ as randomly chosen columns of a Hadamard matrix; 
and $\wbf=\zero$.
Figure~\ref{fig:calibration} plots the success rate versus $L$ and $K$,
where ``success'' is defined as 
$\Exp\|\cbfhat\bbfhat\tran-\cbf^0(\bbf^0)\tran\|_F^2/
 \Exp\|\cbf^0(\bbf^0)\tran\|_F^2 <-60$~dB.
The figure shows that, relative to SparseLift, lifted VAMP gives successful recoveries for a wider range of $L$ and $K$.

\section{Conclusions}  We have extended the analysis of the method
in \cite{rangan2016vamp} to a class of non-separable denoisers.
The method provides a computational efficient method for reconstruction
where structural information and constraints on the unknown vector can
be incorporated in a modular manner.  Importantly, the method admits
a rigorous analysis that can provide precise predictions on the performance
in high-dimensional random settings.

\bibliographystyle{IEEEtran}
\bibliography{../bibl}

\vfill

\pagebreak

{\bf \Large Supplementary Material}

\appendix

\section{Details on Example Denoisers} \label{sec:denoiser_ex_det}

In this section, we provide more details on the denoiser
examples in Section~\ref{sec:denoiser_ex}.
We also provide conditions under which these denoisers satisfy
Assumption~\ref{as:denoiser}.

\paragraph*{Separable Denoisers.}
Assume that $\phi(\cdot)$ satisfies a Lipschitz condition,
\beq \label{eq:philip}
    |\phi(r_2,\gamma_2)-\phi(r_1,\gamma_1)| \leq (A + B|\gamma_2-\gamma_1|)|r_2-r_1| + C|\gamma_2-\gamma_1|,
\eeq
for some constants $A$, $B$ and $C > 0$.
Applying the triangle inequality to \eqref{eq:philip} shows that
$\gbf_1(\cdot)$ satisfies \eqref{eq:lipcond}.
Therefore, $\gbf_1(\cdot)$ satisfies the condition in Definition~\ref{def:uniflip}.
Also, the first limit in \eqref{eq:convlim} is given by,
\[
    \frac{1}{N} \sum_{n=1}^N \phi(x^0_n + z_{1n},\gamma_1)\phi(x^0_n + z_{2n},\gamma_2) =
    \Exp\left[ \phi(x^0_n + z_{1n},\gamma_1)\phi(x^0_n + z_{2n},\gamma_2)\right],
\]
which follows from the Strong Law of Large Numbers
and the fact that we have assumed that the components of $\xbf^0$ are i.i.d.
The remaining limits in\eqref{eq:convlim} can be shown to similar converge.
In particular,
\begin{align*}
   & \lim_{N \arr \infty} \frac{1}{N}\gbf_1(\xbf^0+\zbf_1,\gamma_1)\tran \zbf_2
    = \Exp\left[ \phi(x^0_n + z_{1n},\gamma_1)z_{2n}\right], \\
   & \lim_{N \arr \infty} \bkt{ \nabla \gbf_1(\xbf^0+\zbf_1,\gamma_1) }
    = \Exp\left[ \phi'(x^0_n + z_{1n},\gamma_1) \right],
\end{align*}
where $\phi'(\cdot)$ is the derivative with respect to the first argument.
Moreover, from Stein's lemma,
\[
    \Exp\left[ \phi(x^0_n + z_{1n},\gamma_1)z_{2n}\right] =
    \Exp\left[ \phi'(x^0_n + z_{1n},\gamma_1) \right] \Exp\left[ z_{1n}z_{2n} \right]
    = \Exp\left[ \phi'(x^0_n + z_{1n},\gamma_1) \right] S_{12},
\]
which shows that equality of the two limits in \eqref{eq:divlim}.
This shows that separable, uniform Lipschitz denoisers $\gbf_1(\cdot)$
with i.i.d.\ true signal $\xbf^0$ satisfy Assumption~\ref{as:denoiser}.

\paragraph*{Group-based Denoisers.}
For the groupwise denoiser case, assume that $\phibf(\cdot)$ in
\eqref{eq:ggroup} satisfies a Lipschitz condition,
\beq \label{eq:philipgrp}
    \|\phibf(\rbf_2,\gamma_2)-\phi(\rbf_1,\gamma_1)\|
    \leq (A + B|\gamma_2-\gamma_1|)\|\rbf_2-\rbf_1\| + C|\gamma_2-\gamma_1|,
\eeq
for constants $A,B,C > 0$ and any $\rbf_1,\rbf_2 \in \R^K$.
Then, it is easily verified that $\gbf_1(\cdot)$ is uniformly Lipschitz
according to Definition~\ref{def:uniflip}.
To prove that the denoiser satisfies the convergent conditions in Definition~\ref{def:conv},
let $\zbf_1$ and $\zbf_2$ be two sequence of vectors as in Definition~\ref{def:conv}.
Each $\zbf_i$ can be viewed as a $L \x K$ matrix.  We
let $\zbf_{i\ell}$  be the $\ell$-th row of $\zbf_i$.
With these definitions, the first sum in \eqref{eq:convlim} is given by,
\[
    \frac{1}{LK} \sum_{\ell=1}^L \phibf(\xbf^0_\ell+ \zbf_{1\ell})\phibf(\xbf^0_\ell+ \zbf_{2\ell}),
\]
which is a sum of i.i.d.\ terms.  Hence, the sum converges as $L \arr \infty$.  The convergence
of the other sums can be proven similarly.

\paragraph*{Convolutional Denoisers.}
To prove that $\gbf_1(\rbf_1)$ in \eqref{eq:gconvolution} satisfies Assumption~\ref{as:denoiser},
first observe that since $\gbf_1(\rbf)$ is linear.  Moreover, since it is realized from
a truncated linear filter, its norm is given by,
\[
    \|\gbf_1(\rbf)\| \leq A\|\rbf\|, \quad A := \argmax_{\theta \in [0,2\pi]}
        |\widehat{H}(e^{i\theta})|,
\]
where $\widehat{H}(e^{i\theta})$ is the discrete-time Fourier transform
of the filter $\hbf$.  The bound here holds for all $N$.
Since there is no dependence on $\gamma$, the sequence $\gbf_1(\cdot)$ is uniformly Lipschitz
and satisfies Definition~\ref{def:uniflip}.  To prove that $\xbf^0$ and $\gbf_1(\cdot)$
satisfy Definition~\ref{def:conv},
consider two sequences $\zbf_1$ and $\zbf_2$ be as in Definition~\ref{def:conv}.
Let $\ybf_i$ be the outputs of the convolution (without truncation),
\[
    \ybf_i = \hbf*(\xbf^0+\zbf_i), \quad i=1,2.
\]
Let $\ybf$ and $\zbf$ denote the vector-valued process with components
$(y_{1n},y_{2n})$ and $(z_{1n},z_{2n})$.  By assumption, $\zbf$ is i.i.d.\ Gaussian.
Since $\xbf^0$ is stationary and ergodic and $\zbf$ is i.i.d. and $\hbf$ is a
finite length filter, $\ybf$ is a stationary and ergodic. In addition, $\ybf$ will have
bounded second moments.
Now,
\[
    \gbf_1(\xbf^0+\zbf_i) = T_N(\ybf_i),
\]
which is the first $N$ samples of $\ybf_i$.  Hence,
\[
    \lim_{N \arr \infty} \frac{1}{N} \gbf_1(\xbf^0+\zbf_1)\tran\gbf_1(\xbf^0+\zbf_2)
    = \lim_{N \arr \infty} \frac{1}{N} \sum_{n=0}^{N-1} y_{1n}y_{2n}
\]
and this limit converges almost surely due to the ergodicity of $\ybf$.
The other limits in Definition~\ref{def:conv} can be similarly proven to converge.

\paragraph*{Convolutional Neural Networks.}
As a simple model for a convolutional neural network denoiser,
suppose that the true signal, $\xbf^0$, arises from
$N$ time samples of a stationary and ergodic
multi-variate process $\xbf^0$.  Let $\xbf^0_n \in R^{d_0}$ denote the $n$-th sample
of the process and $d_0$ denote the dimension of the input.
Given an $N$-sample input $\rbf$, if we let
\[
    \zbf_{\lp1} = F_\ell(\zbf_\ell), \quad \zbf_0 = \rbf,
\]
then the $\zbf_L = \gbf_1(\rbf)$.  Assume that each layer output $\zbf_\ell$
has $N$ time samples with dimension $d_\ell$ at each time sample.
Also, assume that each layer $F_\ell(\cdot)$ of the denoiser in \eqref{eq:gcnn}
is one of two possibilities:
\begin{itemize}
\item \emph{Convolutional layer:}  In this case, the layer mapping
$z_{\lp1} = F_\ell(z_\ell)$ is given by a linear multi-channel convolution,
\[
    \zbf_{\ell+1,n} = \sum_{k=0}^{K_\ell-1} \Hbf_{\ell,k}\zbf_{\ell,n-k}, \quad n=0,\ldots,N-1,
\]
where $\Hbf_{\ell,k}$ are the matrix coefficients in a convolution kernel.
We assume the convolution filter are fixed with finite length.

\item \emph{Separable activation:}  In this case, the layer mapping is given by
\[
    \zbf_{\lp1} = \phi_\ell( \zbf_\ell ),
\]
where $\phi_\ell(\cdot)$ is separable and Lipschitz.  This model would include
most common activation functions including sigmoids and ReLUs.
\end{itemize}

Since the convolutional kernels are finite in length, the convolution layers are Lipschitz.
In fact, the Lipschitz constant is given by the spectral norm,
\[
    \|F_\ell(\zbf_\ell)\| \leq A_\ell \|\zbf_\ell\|, \quad
    A_\ell := \max_{\theta \in [0,2\pi]} \sigma_{\rm max}(\widehat{\Hbf}_\ell(e^{i\theta})),
\]
where $\widehat{\Hbf}_\ell(e^{i\theta})$ is the discrete-time multivariable Fourier transform
of the convolution kernel $\Hbf_\ell$ and $\sigma_{\rm max}(\cdot)$ is the maximum
singular value.  By assumption, the activation layers are also Lipschitz.
Since the composition of Lipschitz functions is Lipschitz, the mapping $\gbf_1(\cdot)$
in \eqref{eq:gcnn} is Lipschitz and satisfies Definition~\ref{def:uniflip}.

Also, since $\xbf^0$ is a multi-variate stationary and ergodic random process,
similar arguments as in the convolutional example can be used to show that the limits in
\eqref{eq:convlim} hold almost surely.  Thus, the $\gbf_1(\cdot)$ satisfies
Assumption~\ref{as:denoiser}.

\paragraph*{Singular-Value Thresholding (SVT) Denoiser.} 
To show that $\gbf_1$ in \eqref{eq:svt} is uniformly pseudo-Lipschitz, we first note that $\gbf_1$ is the proximal operator of the nuclear norm $\|\cdot\|_* $, i.e., 
\begin{align*}
    \gbf_1(\rbf,\gamma) = \argmin_{\xbf\in\R^{N_1\times N_2}} \gamma\|\xbf\|_* + \frac{1}{2}\|\xbf-\rbf\|_F^2. 
\end{align*}
From~\cite{candes2013unbiased}, we have that $\gbf_1$ is non-expansive because the nuclear norm is convex and proper, i.e.,
\begin{align}
    \|\gbf_1(\rbf_1,\gamma)-\gbf_1(\rbf_2,\gamma)\|_F^2 &\leq( \rbf_1-\rbf_2)\tran(\gbf_1(\rbf_1,\gamma)-\gbf_1(\rbf_2,\gamma)) \nonumber\\
    \Rightarrow\|\gbf_1(\rbf_1,\gamma)-\gbf_1(\rbf_2,\gamma)\|_F
    &\leq \|\rbf_1-\rbf_2\|_F
    \label{eq:svt_lipschitz} .
\end{align}
Let the SVD of $\rbf_2\in\R^{N_1\times N_2}$ be $\sum_{i=1}^{\min\{N_1,N_2\}} \sigma_i\ubf_i\vbf_i\tran$. 
We can generalize the Lipschitz condition in \eqref{eq:svt_lipschitz} into 
\begin{align*}
    \|\gbf_1(\rbf_1,\gamma_1)-\gbf_1(\rbf_2,\gamma_2)\|_F
    &= \|\gbf_1(\rbf_1,\gamma_1)-\gbf_1(\rbf_2,\gamma_1) + \gbf_1(\rbf_2,\gamma_1)-\gbf_1(\rbf_2,\gamma_2)\|_F \\
    &\leq \|\gbf_1(\rbf_1,\gamma_1)-\gbf_1(\rbf_2,\gamma_1)\|_F + \|\gbf_1(\rbf_2,\gamma_1)-\gbf_1(\rbf_2,\gamma_2)\|_F \\
    &\leq \|\rbf_1-\rbf_2\|_F + \|\gbf_1(\rbf_2,\gamma_1)-\gbf_1(\rbf_2,\gamma_2)\|_F\\
    &\stackrel{\text{(a)}}{=} \|\rbf_1-\rbf_2\|_F + \left\|\sum_{i=1}^{\min\{N_1,N_2\}}((\sigma_i-\gamma_1)_+-(\sigma_i-\gamma_2)_+)\ubf_i\vbf_i\tran\right\|_F\\
    &\leq \|\rbf_1-\rbf_2\|_F + \sum_{i=1}^{\min\{N_1,N_2\}}|(\sigma_i-\gamma_1)_+-(\sigma_i-\gamma_2)_+|\\
    &\leq \|\rbf_1-\rbf_2\|_F + \min\{N_1,N_2\}|\gamma_1-\gamma_2|\\
    &\stackrel{\text{(b)}}{\leq} \|\rbf_1-\rbf_2\|_F + \sqrt{N}|\gamma_1-\gamma_2|,
\end{align*}
where in (a) we have used the the definition of $\gbf_1$ from \eqref{eq:svt} and the SVD of $\rbf_2$, and in (b) we used $\min\{N_1,N_2\}\leq \sqrt{N_1 N_2} = \sqrt{N}$.
Next, we show that $\gbf_1$ also satisfies the convergence conditions in Definition~\ref{def:conv}. 
Let $\zbf_1$ and $\zbf_2$ be two sequences constructed according to Definition~\ref{def:uniflip} and let $\xbf^0$ be the true signal. 
Assume that
\begin{align}\label{eq:svt_limit}
    \lim_{N\rightarrow\infty}\frac{1}{N}\|\xbf^0\|_F^2
    ~~ \text{and} ~~
    \lim_{N\rightarrow\infty}\frac{1}{N} \zbf_1\tran\xbf^0
    ~~ \text{exist almost surely.} 
\end{align}

If we write $\gbf_1(\rbf,\gamma)=[g_1(\rbf,\gamma),\dotso,g_N(\rbf,\gamma)]\tran$, then the following series converges because it is bounded:
\begin{align*}
    \lim_{N\rightarrow\infty}\frac{1}{N}\sum_{i=1}^N|g_i(\xbf^0+\zbf_1,\gamma_1)g_i(\xbf^0+\zbf_2,\gamma_2)|
    &\leq \lim_{N\rightarrow\infty}\frac{1}{N}\|\gbf_1(\xbf^0+\zbf_1,\gamma_1)\|_F
        \|\gbf_1(\xbf^0+\zbf_2,\gamma_2)\|_F\\
    &\leq \lim_{N\rightarrow\infty}\sqrt{\frac{1}{N}\|\xbf^0+\zbf_1\|_F^2}\sqrt{\frac{1}{N}\|\xbf^0+\zbf_2\|_F^2}\\
    &\stackrel{\text{(a)}}{<}\infty,
\end{align*}
where (a) follows from the assumption~\eqref{eq:svt_limit}. 
Since absolute convergence implies convergence, the following series converges:
\begin{align}\label{eq:svt_svd}
    \lim_{N\rightarrow\infty}\frac{1}{N}\gbf_1(\xbf^0+\zbf_1,\gamma_1)\tran\gbf_1(\xbf^0+\zbf_2,\gamma_2)&=\lim_{N\rightarrow\infty}\frac{1}{N}\sum_{i=1}^Ng_i(\xbf^0+\zbf_1,\gamma_1)g_i(\xbf^0+\zbf_2,\gamma_2).
\end{align}

If we choose the covariance matrix in Definition~\ref{def:uniflip} to be $\Sbf=\left[\begin{smallmatrix}1 & 0\\0 & 0\end{smallmatrix}\right]$, then we get 
\begin{align}\label{eq:svt_matrix}
    \lim_{N\rightarrow\infty}\frac{1}{N}\gbf_1(\xbf^0+\zbf_1,\gamma_1)\tran\xbf^0 &=  \lim_{N\rightarrow\infty}\frac{1}{N}\gbf_1(\xbf^0+\zbf_1,\gamma_1)\tran\gbf_1(\xbf^0+\zbf_2,0).
\end{align}
Thus, ~\eqref{eq:svt_matrix} also converges since it is a special case of~\eqref{eq:svt_svd}.

It can be easily shown that $\frac{1}{N}\zbf_2\tran\gbf_1(\xbf^0+\zbf_1,\gamma_1)$ is uniformly Lipschitz. 
Using \cite[Lemma 23]{berthier2017state} and Stein's Lemma~\cite{Stein:72}, we get
\begin{align}
    \lim_{N\rightarrow\infty}\frac{1}{N}\zbf_2\tran\gbf_1(\xbf^0+\zbf_1,\gamma_1)&\stackrel{\text{P}}{\simeq}\lim_{N\rightarrow\infty}\frac{1}{N}\Exp[\zbf_2\tran\gbf_1(\xbf^0+\zbf_1,\gamma_1)] \nonumber\\
    &=\lim_{N\rightarrow\infty}\frac{S_{12}}{N}\Exp[\nabla\gbf_1(\xbf^0+\zbf_1,\gamma_1)]\label{eq:stein_eq},
\end{align}
where $\stackrel{\text{P}}{\simeq}$ denotes convergence in probability.
To show the final convergence condition in Definition~\ref{def:conv}, let us assume that $\langle\nabla\gbf_1(\rbf,\gamma)\rangle$ is uniformly Lipschitz. 
(We are as yet unable to prove this claim.)
Then we have $\lim_{N\rightarrow\infty}\langle\nabla\gbf_1(\rbf,\gamma)\rangle\stackrel{\text{P}}{\simeq}\lim_{N\rightarrow\infty}\Exp[\langle\nabla\gbf_1(\rbf,\gamma)\rangle]$ using \cite[Lemma 23]{berthier2017state}.
Thus, together with~\eqref{eq:stein_eq}, we get the desired result
\begin{align}
        \lim_{N\rightarrow\infty}\frac{1}{N}\langle\nabla\gbf_1(\xbf^0+\zbf_1,\gamma_1)\rangle &= \lim_{N\rightarrow\infty}\frac{1}{S_{12}N}\zbf_2\tran\gbf_1(\xbf^0+\zbf_1,\gamma_1).  
\end{align}

\section{Preliminary Results} \label{sec:prelim}

Since our proof will follow that of \cite{rangan2016vamp},
we review a few key results from that work that will be used here as well.
The most important provides a characterization of a Haar-distributed
matrix $\Vbf$ under linear constraints.  A similar result was key to
the original analysis of Gaussian matrices in the Bayati-Montanari work
\cite{BayatiM:11}.
Let $\Vbf \in \R^{N \x N}$ be Haar-distributed and suppose we wish to find
the conditional distribution of $\Vbf$ under the event that it satisfies
linear constraints
\beq \label{eq:AVB}
    \Abf = \Vbf\Bbf,
\eeq
for some matrices $\Abf, \Bbf \in \R^{N \x s}$ for some $s$.  Assume $\Abf$ and $\Bbf$
are full column rank (hence $s \leq N$).  Let
\beq \label{eq:UABdef}
    \Ubf_\Abf = \Abf(\Abf\tran\Abf)^{-1/2}, \quad
    \Ubf_\Bbf = \Bbf(\Bbf\tran\Bbf)^{-1/2}.
\eeq
Also, let $\Ubf_{\Abf^\perp}$ and $\Ubf_{\Bbf^\perp}$ be any $N \x (N-s)$
matrices whose columns are
an orthonormal bases for $\mathrm{Range}(\Abf)^\perp$ and
$\mathrm{Range}(\Bbf)^\perp$, respectively.

\begin{lemma} \cite[Lemma 4]{rangan2016vamp} \label{lem:orthogLin}
Let $\Vbf \in \R^{N \x N}$ be a random matrix
that is Haar distributed.  Suppose that $\Abf$ and $\Bbf$ are deterministic and $G$ is
the event that $\Vbf$ satisfies linear constraints \eqref{eq:AVB}.
Then, the can write $\Vbf$ as,
\[
    \Vbf =
    \Abf(\Abf\tran\Abf)^{-1}\Bbf\tran + \Ubf_{\Abf^\perp}\tilde{\Vbf}\Ubf_{\Bbf^\perp}\tran,
\]
where $\tilde{\Vbf}$ is a Haar distributed matrix independent of $G$.
\end{lemma}

Lemma~\ref{lem:orthogLin} is used in conjunction with the following
result.

\begin{lemma} \label{lem:orthogGaussLim}
Fix a dimension $s \geq 0$,
and suppose that we have sequences
$\xbf = \xbf(N)$ and $\Ubf=\Ubf(N)$ are sequences such that
for each $N$,
\begin{enumerate}[(i)]
\item $\Ubf=\Ubf(N) \in \R^{N\x (N-s)}$ is a random
matrix with $\Ubf\tran\Ubf = \Ibf$;
\item $\xbf=\xbf(N) \in \R^{N-s}$ a random vector whose mean squared magnitude
converges almost surely as
\[
    \lim_{N \arr \infty} \frac{1}{N} \|\xbf\|^2 = \tau,
\]
for some $\tau > 0$.

\item $\Vbf=\Vbf(N) \in \R^{(N-s) \x (N-s)}$ is a Haar distributed, independent
of $\Ubf$ and $\xbf$.
\end{enumerate}
Then, if we define $\ybf = \Ubf\Vbf\xbf$, we have that the components of $\ybf$
are approximately Gaussian in that,
\beq \label{eq:ygausslim}
    \ybf = \tilde{\ybf} + \etabf,
\eeq
where $\tilde{\ybf} \sim \Norm(0,\tau \Ibf)$ and
\[
    \lim_{N \arr \infty} \frac{1}{N} \|\etabf\|^2 = 0,
\]
almost surely.
\end{lemma}
\begin{proof}  This can be proven similar to that of \cite[Lemma 5]{rangan2016vamp}.
\end{proof}

\section{A General Convergence Result}

Similar to the proof in \cite{rangan2017vamp},
we prove our main result, Theorem~\ref{thm:se}, by considering
the following more general recursion.  We are given a dimension $N$,
an orthogonal matrix $\Vbf \in \R^{N \x N}$,
an initial random vector $\ubf_0 \in \R^N$, along with random vectors
$\wbf^p, \wbf^q \in \R^N$.
Then, we generate a sequence of iterates by the following recursion:%
\begin{subequations}
\label{eq:algoGen}
\begin{align}
    \pbf_k &= \Vbf\ubf_k \label{eq:pupgen} \\
    \alpha_{1k} &= \bkt{ \nabla \fbf_p(\pbf_k,\wbf^p,\gamma_{1k})},
    \quad \gamma_{2k} = \Gamma_1(\gamma_{1k},\alpha_{1k})
        \label{eq:alpha1gen}  \\
    \vbf_k &= C_1(\alpha_{1k})\left[
        \fbf_p(\pbf_k,\wbf^p,\gamma_{1k})- \alpha_{1k} \pbf_{k} \right] \label{eq:vupgen} \\
    \qbf_k &= \Vbf\tran\vbf_k \label{eq:qupgen} \\
    \alpha_{2k} &= \bkt{ \nabla \fbf_q(\qbf_k,\wbf^q,\gamma_{2k})},
    \quad \gamma_{1,\kp1} =\Gamma_2(\gamma_{2k},\alpha_{2k})
        \label{eq:alpha2gen} \\
    \ubf_{\kp1} &= C_2(\alpha_{2k})\left[
        \fbf_q(\qbf_k,\wbf^q,\gamma_{2k}) - \alpha_{2k}\qbf_{k} \right], \label{eq:uupgen}
\end{align}
\end{subequations}
which is initialized with $\ubf_0$ and a scalar $\gamma_{10}$.
We index the recursions  by $N$.
We assume that the initial constant and norm of the initial vector converges as
\beq \label{eq:gam10limgen}
    \lim_{N \arr \infty} \gamma_{10} = \gammabar_{10}, \quad \lim_{N \arr \infty} \frac{1}{N} \|\ubf_0\|^2 = \tau_{10},
\eeq
for some constants $\gammabar_{10}$ and $\tau_{10}$.
The matrix $\Vbf \in \R^{N \x N}$
is assumed to be uniformly distributed on the set of orthogonal matrices
independent of $\ubf_0$, $\wbf^p$ and $\wbf^q$.
For the functions $\fbf_p(\cdot)$ and $\fbf_q(\cdot)$ we need a slight generalization
of Definitions~\ref{def:uniflip} and \ref{def:conv}.

\begin{definition}  \label{def:uniflipgen}
For each $N$, suppose that
 $\ubf \in \R^N$ is a random vector
and $\fbf(\zbf,\ubf,\gamma) \in \R^N$ is a function on $\zbf\in \R^N$, $\ubf \in \R^N$
and $\gamma \in \R$.  Let $G$ be some closed, convex set of values $\gamma$.
We say the sequence is
\emph{uniformly Lipschitz continuous} if
there exists constants $A$, $B$ and $C > 0$, such that
\begin{align}
    \MoveEqLeft
    \limsup_{N \arr \infty} \lim_{N \arr \infty} \frac{1}{\sqrt{N}}\|\fbf(\zbf_2,\ubf,\gamma_2)-
        \fbf(\zbf_1,\ubf,\gamma_1)\| \\
    &\leq
    \limsup_{N \arr \infty} \frac{A + B|\gamma_2-\gamma_1|}{\sqrt{N}}\|\zbf_2-\zbf_1\| + C|\gamma_2-\gamma_1|,
    \label{eq:lipcondgen}
\end{align}
almost surely  for any $\zbf_1,\zbf_2$ and $\gamma_1,\gamma_2 \in G$.
\end{definition}

\begin{definition}  \label{def:convgen}
Let $\ubf$, $\fbf(\cdot)$ and $G$ be as in Definition~\ref{def:uniflipgen}.
The sequence $\ubf$ and $\fbf(\cdot)$ are said to
be \emph{convergent under Gaussian noise} if the following condition holds:
Let $\zbf_1,\zbf_2 \in \R^N$ be two sequences
where $(z_{n1},z_{n2})$ are i.i.d.\ with  $(z_{n1},z_{n2})=\Norm(0,\Sbf)$
for some positive definite covariance $\Sbf \in \R^{2 \x 2}$.
Then, the following limits exists almost surely,
\begin{align}
    {\mathcal M}(\Sbf,\gamma_1,\gamma_2) &:=
    \lim_{N \arr \infty} \fbf(\ubf,\zbf_1,\gamma_1)\tran \fbf(\ubf,\zbf_2,\gamma_2)
    \label{eq:momlimgen} \\
    {\mathcal A}(S_{11},\gamma_1) &:= \lim_{N \arr \infty}
        \bkt{\nabla \fbf(\ubf,\zbf_1,\gamma_1)}
        = \frac{1}{NS_{12}} \fbf(\ubf,\zbf_1,\gamma_1)\tran \zbf_2,
    \label{eq:divlimgen}
\end{align}
for all $\gamma_1,\gamma_2 \in G$ and covariance matrices $\Sbf$.
Moreover, the functions ${\mathcal M}(\cdot)$ and ${\mathcal A}(\cdot)$ are
continuous in $\Sbf$, $\gamma_1$ and $\gamma_2$.
\end{definition}

\medskip
Our critical assumption is that, following Definitions~\ref{def:uniflipgen} and
\ref{def:convgen},
the sequence of random vectors $\wbf^p$ and functions $\fbf_p(\pbf,\wbf^p,\gamma_1)$ (as indexed by $N$)
are uniformly Lipschitz continuous and convergent under Gaussian noise for $\gamma_1 \in G_1$
for some closed, convex set $G_1$.
Similarly, the sequence $\wbf^q$ and function $\fbf_q(\pbf,\wbf^q,\gamma_2)$ is also
uniformly Lipschitz continuous and convergent under Gaussian noise
for $\gamma_2 \in G_2$ for some closed, convex set $G_2$.
In this case, we can define the second moments,
\begin{subequations} \label{eq:momentpq}
\begin{align}
  {\mathcal M}_p(\tau_1,\gamma_1) &:= \lim_{N \arr \infty} \frac{1}{N} \|\fbf_p(\pbf,\wbf^p,\gamma_1)\|^2, \quad \pbf \sim \Norm(0,\tau_1\Ibf), \\
  {\mathcal M}_q(\tau_2,\gamma_2) &:= \lim_{N \arr \infty} \frac{1}{N} \|\fbf_q(\qbf,\wbf^q,\gamma_2)\|^2, \quad \qbf \sim \Norm(0,\tau_2\Ibf),
\end{align}
\end{subequations}
as well as the sensitivity functions,
\begin{subequations} \label{eq:divpq}
\begin{align}
  {\mathcal A}_p(\tau_1,\gamma_1) &:= \lim_{N \arr \infty} \bkt{ \nabla \fbf_p(\pbf,\wbf^p,\gamma_1)}), \quad \pbf \sim \Norm(0,\tau_1\Ibf), \\
  {\mathcal A}_q(\tau_2,\gamma_2) &:= \lim_{N \arr \infty} \bkt{ \nabla \fbf_q(\qbf,\wbf^q,\gamma_2)},
  \quad \qbf \sim \Norm(0,\tau_2\Ibf).
\end{align}
\end{subequations}
The limits exist due to the assumption of $\fbf_p$ and $\fbf_q$ being convergent under Gaussian noise.
In addition, Definition~\ref{def:convgen} shows that the sensitivity functions are also given by,
\begin{subequations} \label{eq:divpqstein}
\begin{align}
  {\mathcal A}_p(\tau_1,\gamma_1) &=  = \lim_{N \arr \infty} \frac{1}{N\tau_1}
  \pbf\tran\fbf_p(\pbf,\wbf^p,\gamma_1), \quad \pbf \sim \Norm(0,\tau_1\Ibf), \\
  {\mathcal A}_q(\tau_2,\gamma_2) &=   \lim_{N \arr \infty} \frac{1}{N\tau_2}
  \qbf\tran\fbf_q(\qbf,\wbf^q,\gamma_2)
  \quad \qbf \sim \Norm(0,\tau_2\Ibf).
\end{align}
\end{subequations}

Under the above assumptions, define the SE equations,
\begin{subequations} \label{eq:segen}
\begin{align}
    \alphabar_{1k} &= {\mathcal A}_p(\tau_{1k}, \gammabar_{1k})  \label{eq:a1segen} \\
    \tau_{2k} &= C_1^2(\alphabar_{1k}) \left\{  {\mathcal M}_p(\tau_{1k}, \gammabar_{1k})  - \alphabar_{1k}^2\tau_{1k} \right\}
        \label{eq:tau2segen} \\
    \gammabar_{2k} &= \Gamma_1(\gammabar_{1k},\alphabar_{1k}) \label{eq:gam2segen} \\
    \alphabar_{2k} &= {\mathcal A}_q(\tau_{2k}, \gammabar_{2k}) \label{eq:a2segen} \\
    \tau_{1,\kp1} &= C_2^2(\alphabar_{2k})\left\{
        {\mathcal M}_p(\tau_{2k}, \gammabar_{2k})  - \alphabar_{2k}^2\tau_{2k}\right\}
        \label{eq:tau1segen} \\
    \gamma_{1,\kp1} &= \Gamma_2(\gammabar_{2k},\alphabar_{2k}), \label{eq:gam1segen}
\end{align}
\end{subequations}
which are initialized with $\gammabar_{10}$ and $\tau_{10}$ in \eqref{eq:gam10limgen}.

For the sequel, we will use the notation that, if $\xbf=\xbf(N)$ and $\ybf=\ybf(N) \in \R^N$
are two sequences of random vectors that scale with $N$,
\beq \label{eq:orderN}
    \xbf = \ybf + \ONinv \Longleftrightarrow
    \lim_{N \arr \infty} \frac{1}{N}\|\xbf-\ybf\|^2 \mbox{ almost surely.}
\eeq
With this definition, we have the following result.

\begin{theorem} \label{thm:genConv}  Consider the recursions \eqref{eq:algoGen}
and SE equations \eqref{eq:segen} under the above assumptions.  Assume additionally that,
for all $k$ and $i=1,2$,
the functions $C_i(\alpha_i)$ and $\Gamma_i(\gamma_i,\alpha_i)$
are continuous at the points $(\gamma_i,\alpha_i)=(\gammabar_{ik},\alphabar_{ik})$
from the SE equations. Also, assume that $\gammabar_{ik} \in G_i$ for all $i$. Then,
\begin{enumerate}[(a)]
\item For each $k$, we can write $\pbf_k = \tilde{\pbf}_k + \ONinv$ such that the
matrix,
\beq \label{eq:Pkgen}
    \tilde{\Pbf}_k = [\tilde{\pbf}_0, \cdots, \tilde{\pbf}_k] \in \R^{N \x {\kp1}},
\eeq
is independent of $\wbf^p$ and has i.i.d.\ rows, $(\tilde{p}_{n0},\cdots,\tilde{p}_{nk})$,
that are zero mean, $\kp1$-dimensional Gaussian random vectors.
In addition, we have that
\beq \label{eq:ag1limgen}
    E\tilde{p}^2_{nk} = \tau_{1k}, \quad
    \lim_{N \arr \infty} (\alpha_{1k},\gamma_{2k}) = (\alphabar_{1k},\gammabar_{2k}),
\eeq
where the limit holds almost surely.

\item For each $k$, we can write $\qbf_k = \tilde{\qbf}_k + \ONinv$ such that the
matrix,
\beq \label{eq:Qkgen}
    \tilde{\Qbf}_k = [\tilde{\qbf}_0, \cdots, \tilde{\qbf}_k] \in \R^{N \x {\kp1}},
\eeq
is independent of $\wbf^q$ and has i.i.d.\ rows, $(\tilde{q}_{n0},\cdots,\tilde{q}_{nk})$,
that are zero mean, $\kp1$-dimensional Gaussian random vectors.
In addition, we have that
\beq \label{eq:ag2limgen}
    E\tilde{q}^2_{nk} = \tau_{2k}, \quad
    \lim_{N \arr \infty} (\alpha_{2k},\gamma_{1,\kp1}) = (\alphabar_{2k},\gammabar_{1,\kp1}),
\eeq
where the limit holds almost surely.
\end{enumerate}
\end{theorem}
\begin{proof}  We will prove this in the next Appendix,
Appendix~\ref{sec:genConvPf}.
\end{proof}

\section{Proof of Theorem~\ref{thm:genConv}} \label{sec:genConvPf}

\subsection{Induction Argument}

The proof has a similar structure to the proof of the general convergence
result in \cite{rangan2016vamp}.  So, we will highlight only the key differences.
Similar to \cite{rangan2016vamp}, we use
an induction argument.   Given iterations $k, \ell \geq 0$,
define the hypothesis, $H_{k,\ell}$ as the statement:
\begin{itemize}
\item Part (a) of Theorem~\ref{thm:genConv} is true up to $k$; and
\item Part (b) of Theorem~\ref{thm:genConv} is true up to $\ell$.
\end{itemize}
The induction argument will then follow by showing the following three facts:
\begin{itemize}
\item $H_{0,-1}$ is true;
\item If $H_{k,\km1}$ is true, then so is $H_{k,k}$;
\item If $H_{k,k}$ is true, then so is $H_{\kp1,k}$.
\end{itemize}

\subsection{Induction Initialization}
We first show that the hypothesis $H_{0,-1}$ is true.  That is,
we must show that the rows of \eqref{eq:Pkgen} are i.i.d.\ Gaussians
and the limits in \eqref{eq:ag1limgen} hold for $k=0$.
This is a special case of Lemma~\ref{lem:orthogGaussLim}.
Specifically, for each $N$, let $\Ubf=\Ibf_N$, the $N \x N$ identity
matrix, which trivially satisfies property (i) of Lemma~~\ref{lem:orthogGaussLim}
with $s=0$.  Also, $\xbf = \ubf_0$ satisfies property (ii) due to the assumption
\eqref{eq:gam10limgen}.  Then, since
$\pbf_0 = \Vbf\ubf_0=\Ubf\Vbf\xbf$ and $\Vbf$ is
Haar distributed independent of $\ubf_0$, we have that
\beq \label{eq:ptilde0}
    \pbf_0 = \tilde{\pbf}_0 + O(\tfrac{1}{N}), \quad \tilde{\pbf}_0 \sim \Norm(0,\tau_{10}\Ibf).
\eeq
This  proves the Gaussianity of the rows of \eqref{eq:Pkgen} for $k=0$.
Also,
\begin{align}
    \lim_{N \arr \infty} \alpha_{10} &\stackrel{(a)}{=}
    \lim_{N \arr \infty} \bkt{ \nabla \fbf_p(\pbf_0,\wbf^p,\gamma_{10}) }  \nonumber \\
    & \stackrel{(b)}{=}
    \lim_{N \arr \infty} \bkt{ \nabla \fbf_p(\tilde{\pbf}_0,\wbf^p,\gammabar_{10}) }
    \stackrel{(c)}{=} {\mathcal A}_p(\tau_{10},\gammabar_{10})
    \stackrel{(d)}{=} \alphabar_{10},     \label{eq:alpha0convgen}
\end{align}
where (a) follows from  \eqref{eq:alpha1gen};
(b) follows from \eqref{eq:gam10limgen}, \eqref{eq:ptilde0} along with the Lipschitz
continuity assumption of $\fbf_p(\cdot)$;
(c) follows from the definition \eqref{eq:divpq}; and
(d) follows from \eqref{eq:a1segen}.  In addition,
\beq \label{eq:gam0convgen}
    \lim_{N \arr \infty} \gamma_{10} \stackrel{(a)}{=} \lim_{N \arr \infty}
    \Gamma_1(\gamma_{10},\alpha_{10})
    \stackrel{(b)}{=} \Gamma_1(\gammabar_{10},\alphabar_{10})
    \stackrel{(c)}{=} \gammabar_{20}
\eeq
where (a) follows from \eqref{eq:alpha1gen};
(b) follows from \eqref{eq:gam10limgen}, \eqref{eq:alpha0convgen} and the continuity
of $\Gamma_1(\cdot)$ and (c) follows from \eqref{eq:gam2segen}.
This proves \eqref{eq:ag1limgen}.

\subsection{The Induction Recursion}
We next show the implication $H_{k,\km1} \Rightarrow H_{k,k}$. The implication
$H_{k,k} \Rightarrow H_{\kp1,k}$ is proven similarly.  Hence, fix $k$ and assume
that $H_{k,\km1}$ holds.  To show $H_{k,k}$, we need to show the Gaussianity of the rows of
\eqref{eq:Qkgen} and that the limits in \eqref{eq:ag2limgen} hold.

First, similar to the proof of \eqref{eq:alpha0convgen}, we have that
\beq
    \lim_{N \arr \infty} \alpha_{2k} = \alphabar_{2k}.     \label{eq:alphakconvgen}
\eeq
Also, by the induction hypothesis, $\gamma_{2k} \arr \gammabar_{2k}$, and similar
to the proof of \eqref{eq:gam0convgen},
\beq
    \lim_{N \arr \infty} \gamma_{1,\kp1} = \gammabar_{1,\kp1}.     \label{eq:gam2kconvgen}
\eeq
This proves \eqref{eq:ag2limgen}.  We next need to compute various correlations.

\begin{lemma} \label{lem:vpcorr}
Under the hypothesis $H_{k,\km1}$, then
for any $i,j=0,\ldots,k$ the following limits exist almost surely,
\beq \label{eq:vpcorr1}
    \lim_{N \arr \infty} \frac{1}{N} \pbf_i\tran\pbf_j, \quad
    \lim_{N \arr \infty} \frac{1}{N} \vbf_i\tran\vbf_j.
\eeq
Also,
\beq \label{eq:vpcorr2}
    \lim_{N \arr \infty} \frac{1}{N} \|\vbf_k\|^2 = \tau_{2k}, \quad
    \lim_{N \arr \infty} \frac{1}{N} \vbf_i\tran\pbf_j = 0.
\eeq
\end{lemma}
\begin{proof}
For the first part of \eqref{eq:vpcorr1},
\[
    \lim_{N \arr \infty} \frac{1}{N} \pbf_i\tran\pbf_j
        \stackrel{(a)}{=}\lim_{N \arr \infty} \frac{1}{N} \tilde{\pbf}_i\tran\tilde{\pbf}_j
        \stackrel{(b)}{=} \Exp(\tilde{p}_{in}\tilde{p}_{jn}),
\]
where (a) follows due to induction hypothesis that
$\pbf_\ell = \tilde{\pbf}_\ell + O(\tfrac{1}{N})$ for $\ell \leq k$
and (b) follows from the fact that $(\tilde{p}_{in},\tilde{p}_{jn})$
are i.i.d.,
so the limit occurs almost surely by the Strong Law of Large Numbers.
For the second part of \eqref{eq:vpcorr1},
\begin{align}
  \MoveEqLeft \lim_{N \arr \infty} \frac{1}{N} \vbf_i\tran\vbf_j \nonumber \\
    &\stackrel{(a)}{=} \lim_{N \arr \infty} \frac{C_1(\alpha_{1i})C_1(\alpha_{1j})}{N}
    \left[ \fbf_p(\pbf_i,\wbf^p,\gamma_{1i})-\alpha_{1i}\pbf_i \right]\tran
    \left[ \fbf_p(\pbf_j,\wbf^p,\gamma_{1j})-\alpha_{1j}\pbf_j \right] \nonumber \\
    &\stackrel{(b)}{=} \lim_{N \arr \infty} \frac{C_1(\alphabar_{1i})C_1(\alphabar_{1j})}{N}
    \left[ \fbf_p(\tilde{\pbf}_i,\wbf^p,\gammabar_{1i})-\alphabar_{1i}\tilde{\pbf}_i \right]\tran
    \left[ \fbf_p(\tilde{\pbf}_j,\wbf^p,\gammabar_{1j})-\alphabar_{1j}\tilde{\pbf}_j \right],
\end{align}
where (a) follows from \eqref{eq:vupgen};
(b) follows from the fact that $\pbf_k = \tilde{\pbf}_k + \ONinv$,
\eqref{eq:ag2limgen} and the continuity assumptions of $\fbf_p(\cdot)$ and $C_1(\cdot)$.
We can expand this sum into four terms and use the fact that $\fbf_p(\cdot)$ is convergent under Gaussian
noise to show that all the terms are converge almost surely.
Hence, both the limits in \eqref{eq:vpcorr1} exist almost surely.

In the special case when $i=j=k$, we have that,
\begin{align}
  \MoveEqLeft \lim_{N \arr \infty} \frac{1}{N} \|\vbf_k\|^2
    = \lim_{N \arr \infty} \frac{C_1^2(\alphabar_{1i})}{N}
    \| \fbf_p(\tilde{\pbf}_k,\wbf^p,\gammabar_{ki})-\alphabar_{1i}\tilde{\pbf}_k\|^2
    \nonumber \\
    & = \lim_{N \arr \infty} \frac{C_1^2(\alphabar_{1k})}{N}
    \left[ \|\fbf_p(\tilde{\pbf}_k,\wbf^p,\gammabar_{1k})\|^2
    - 2 \alphabar_{1k}\tilde{\pbf}_k\tran \fbf_p(\tilde{\pbf}_k,\wbf^p,\gammabar_{1k})
    + \alphabar_{1k}^2 \|\tilde{\pbf}_k \|^2 \right] \nonumber \\
    & \stackrel{(a)}{=} C_1^2(\alphabar_{1k}) \left(
        {\mathcal M}_p(\tau_{1k},\gammabar_{1k})
    - 2 \alphabar_{1k}^2 \tau_{1k}
    + \alphabar_{1k}^2 \tau_{1k} \right) \nonumber \\
    & = C_1^2(\alphabar_{1k}) \left(
        {\mathcal M}_p(\tau_{1k},\gammabar_{1k})
    - \alphabar_{1k}^2 \tau_{1k}
     \right) \stackrel{(b)}{=} \tau_{2k},
\end{align}
where (a) follows from the limits in \eqref{eq:momentpq} and
\eqref{eq:divpqstein}; and
(b) follows from \eqref{eq:tau2segen}.  This proves the first relation in
\eqref{eq:vpcorr2}.  For the second relation,
\begin{align}
  \MoveEqLeft \lim_{N \arr \infty} \frac{1}{N} \vbf_i\tran\pbf_j
  \stackrel{(a)}{=} \lim_{N \arr \infty} \frac{C_1(\alpha_{1i})}{N}
    \left( \fbf_p(\pbf_k,\wbf^p,\gamma_{1k})-\alpha_{1i}\pbf_i
    \right)\tran\pbf_j  \nonumber \\
    & \stackrel{(b)}{=} \lim_{N \arr \infty} \frac{C_1(\alphabar_{1i})}{N}
    \left( \fbf_p(\tilde{\pbf}_k,\wbf^p,\gammabar_{1k})-\alphabar_{1i}\tilde{\pbf}_i
    \right)\tran\tilde{\pbf}_j  \nonumber \\
   &\stackrel{(c)}{=} ({\mathcal A}_p(\tau_{1i},\gammabar_{1i}) - \alphabar_{1i})
    \mathrm{cov}(\tilde{p}_{ni}, \tilde{p}_{jn})
    \stackrel{(d)}{=} 0,
\end{align}
where (a) follows from \eqref{eq:vupgen};
(b) follows from the fact that $\pbf_k = \tilde{\pbf}_k + \ONinv$;
(c) follows from the assumption that $\fbf_p(\cdot)$ is
convergent under Gaussian noise as given in Definition~\ref{def:conv};
and (d) follows from \eqref{eq:a1segen}.
\end{proof}

The remainder of the proof now follows a very similar structure to that
in \cite{rangan2016vamp}.
First, let
\[
    \Ubf_k := \left[ \ubf_0 \cdots \ubf_k \right] \in \R^{N \x (\kp1)},
\]
represent the first $\kp1$ values of the vector $\ubf_\ell$.
Define the matrices $\Vbf_k$, $\Qbf_k$ and $\Pbf_k$ similarly.
Let $G_k$ be the set of random vectors,
\beq \label{eq:Gdef}
    G_k := \left\{ \Ubf_k, \Pbf_k, \Vbf_k, \Qbf_{\km1} \right\}.
\eeq
With some abuse of notation, we will also use $G_k$
to denote the sigma-algebra generated by these variables.
The set \eqref{eq:Gdef} contains all the outputs of the algorithm
\eqref{eq:algoGen} immediately \emph{before} \eqref{eq:qupgen} in iteration $k$.

Now, the actions of the matrix $\Vbf$ in the recursions \eqref{eq:algoGen}
are through the matrix-vector multiplications \eqref{eq:pupgen} and \eqref{eq:qupgen}.
Hence, if we define the matrices,
\beq \label{eq:ABdef}
    \Abf_k := \left[ \Pbf_k ~ \Vbf_{\km1} \right], \quad
    \Bbf_k := \left[ \Ubf_k ~ \Qbf_{\km1} \right],
\eeq
the output of the recursions in the set $G_k$ will be unchanged for all
matrices $\Vbf$ satisfying the linear constraints
\beq \label{eq:ABVconk}
    \Abf_k = \Vbf\Bbf_k.
\eeq
Hence, the conditional distribution of $\Vbf$ given $G_k$ is precisely
the uniform distribution on the set of orthogonal matrices satisfying
\eqref{eq:ABVconk}.  The matrices $\Abf_k$ and $\Bbf_k$ are of dimensions
$N \x s$ where $s=2k+1$.
From Lemma~\ref{lem:orthogLin},
\beq \label{eq:Vconk}
    \Vbf =
    \Abf_k(\Abf\tran_k\Abf_k)^{-1}\Bbf_k\tran + \Ubf_{\Abf_k^\perp}\tilde{\Vbf}\Ubf_{\Bbf_k^\perp}\tran,
\eeq
where $\Ubf_{\Abf_k^\perp}$ and $\Ubf_{\Bbf_k^\perp}$ are $N \x (N-s)$ matrices
whose columns are an orthonormal basis for $\Range(\Abf_k)^\perp$ and $\Range(\Bbf_k)^\perp$.
The matrix $\tilde{\Vbf}$ is  Haar distributed on the set of $(N-s)\x(N-s)$
orthogonal matrices and independent of $G_k$.

Next, similar to the proof in \cite{rangan2016vamp},
we use \eqref{eq:Vconk} to write $\qbf_k$ in \eqref{eq:qupgen} as a sum of two terms
\beq \label{eq:qpart}
    \qbf_k = \Vbf\tran\vbf_k = \qbf_k^{\rm det} + \qbf_k^{\rm ran},
\eeq
where $\qbf_k^{\rm det}$ is what we will call the \emph{deterministic} part:
\beq \label{eq:qkdet}
    \qbf_k^{\rm det} = \Bbf_k(\Abf\tran_k\Abf_k)^{-1}\Abf_k\tran\vbf_k,
\eeq
and $\qbf_k^{\rm ran}$ is what we will call the \emph{random} part:
\beq \label{eq:qkran}
    \qbf_k^{\rm ran} = \Ubf_{\Bbf_k^\perp}\tilde{\Vbf}\tran \Ubf_{\Abf_k^\perp}\tran \vbf_k.
\eeq
The next two lemmas evaluate the asymptotic distributions of the two
terms in \eqref{eq:qpart} and are similar to those in the proof in
\cite{rangan2016vamp}.

\begin{lemma} \label{lem:qconvdet}
Under the induction hypothesis $H_{k,\km1}$, there exists constants
$\beta_{k,0},\ldots,\beta_{k,\km1}$ such that
\beq \label{eq:Qkdetlim}
    \qbf_k^{\rm det} = \beta_{k0}\tilde{\qbf}_0 +
        \cdots + \beta_{k,\km1}\tilde{\qbf}_{\km1} + \ONinv.
\eeq
\end{lemma}
\begin{proof}
From Lemma~\ref{lem:vpcorr}, these exists almost surely.
We evaluate the asymptotic values of various terms in \eqref{eq:qkdet}.
Using the definition of $\Abf_k$ in \eqref{eq:ABdef},
\[
    \Abf\tran_k\Abf_k = \left[ \begin{array}{cc}
        \Pbf_k\tran\Pbf_k & \Pbf_k\tran\Vbf_{\km1} \\
        \Vbf_{\km1}\tran\Pbf_k & \Vbf_{\km1}\tran\Vbf_{\km1}
        \end{array} \right]
\]
For $i,j \leq k$, define
\[
    Q^p_{ij} :=
    \lim_{N \arr \infty} \frac{1}{N} \pbf_i\tran\pbf_j, \quad
    Q^v_{ij} := \lim_{N \arr \infty} \frac{1}{N} \vbf_i\tran\vbf_j.
\]
From Lemma~\ref{lem:vpcorr}, these limits exists almost surely.
Let $\Qbf^p$ be the matrix with components $Q^p_{ij}$ for $i,j\leq k$
and let $\Qbf^v$ be the matrix with components $Q^v_{ij}$ for $i,j< k$.
Then, since $\pbf_i$ and $\pbf_j$ are the $i$-th and $j$-th column of
$\Pbf_k$, the
$(i,j)$-th component of the matrix $\Pbf_k\tran\Pbf_k$ is given by
\begin{align*}
    \lim_{N \arr \infty} \frac{1}{N} \left[ \Pbf_k\tran\Pbf_k \right]_{ij}
        = \lim_{N \arr \infty} \frac{1}{N} \pbf_i\tran\pbf_j = Q^p_{ij}.
\end{align*}
Similarly,
\[
    \lim_{N \arr \infty} \frac{1}{N} \Vbf_{\km1}\tran\Vbf_{\km1} = \Qbf^v
\]
almost surely.
Also, from Lemma~\ref{lem:vpcorr},
\[
    \lim_{N \arr \infty} \frac{1}{N} \Pbf_k\tran\Vbf_{\km1} = 0,
\]
almost surely.
The above calculations show that
\beq \label{eq:AAlim}
    \lim_{N \arr \infty} \frac{1}{N}\Abf\tran_k\Abf_k = \left[ \begin{array}{cc}
        \Qbf^p & \mathbf{0} \\
        \mathbf{0} & \Qbf^v
        \end{array} \right].
\eeq
A similar calculation shows that
\beq \label{eq:Aslim}
    \lim_{N \arr \infty} \frac{1}{N} \Abf\tran_k\vbf_k = \left[
    \begin{array}{c} \mathbf{0} \\ \bbf^v \end{array} \right],
\eeq
where $\bbf^v$ is the vector of correlations
\beq
    \bbf^v = \left[ Q^v_{0k} ~ Q^v_{1k} ~\cdots ~Q^v_{\km1,k} \right]\tran.
\eeq
Combining \eqref{eq:AAlim} and \eqref{eq:Aslim} shows that
\beq \label{eq:Vsmult1}
    \lim_{N \arr \infty} (\Abf\tran_k\Abf_k)^{-1}\Abf_k\tran \vbf_k =
    \left[ \begin{array}{c} \mathbf{0} \\ \mathbf{\beta}_k \end{array} \right],
\eeq
where
\[
    \mathbf{\beta}_k := \left[ \Qbf^v\right]^{-1}\bbf^v.
\]
Therefore,
\begin{align}
    \MoveEqLeft \qbf_k^{\rm det} = \Bbf_k(\Abf\tran_k\Abf_k)^{-1}\Abf_k\tran\vbf_k 
    = \left[ \Ubf_k ~ \Qbf_{\km1} \right]
    \left[ \begin{array}{c} \mathbf{0} \\ \mathbf{\beta}_k \end{array} \right]
    +\ONinv 
    = \sum_{\ell=0}^{\km1} \beta_{k\ell} \tilde{\qbf}_\ell +\ONinv.
\end{align}
This completes the proof of the lemma.
\end{proof}

\begin{lemma} \label{lem:rhoconv}
Under the induction hypothesis $H_{k,\km1}$, the following limit
holds almost surely
\beq
    \lim_{N \arr \infty} \frac{1}{N} \| \Ubf_{\Abf_k^\perp}\tran\vbf_k\|^2 =
    \rho_k,
\eeq
for some constant $\rho_k \geq 0$.
\end{lemma}
\begin{proof}  From \eqref{eq:ABdef}, the matrix  $\Abf_k$ has $s=2k+1$ columns.
From Lemma~\ref{lem:orthogLin},
$\Ubf_{\Abf_k^\perp}$ is an orthonormal basis of $N-s$ in the $\mathrm{Range}(\Abf_k)^\perp$.
Hence, the energy $\| \Ubf_{\Abf_k^\perp}\vbf_k\|^2$ is precisely
\[
    \| \Ubf_{\Abf_k^\perp}\sbf_k\|^2 = \vbf_k\tran\vbf_k - \vbf_k\tran\Abf_k
        (\Abf_k\tran\Abf_k)^{-1}\Abf_k\tran\vbf_k.
\]
Using similar calculations as the previous lemma, we have
\[
    \lim_{N \arr \infty} \frac{1}{N} \| \Ubf_{\Abf_k}\sbf_k\|^2
    = \tau_{2k} - (\bbf^v)\tran\left[ \Qbf^v \right]^{-1}\bbf^v.
\]
Hence, the lemma is proven if we define $\rho_k$ as the right hand side of this
equation.
\end{proof}

\begin{lemma} \label{lem:qconvran}
Under the induction hypothesis $H_{k,\km1}$, the
``random" part $\qbf_k^{\rm ran}$ is given by,
\beq \label{eq:qranlim}
    \qbf_k^{\rm ran} = \ubf_k + \ONinv,
\eeq
where $\ubf_k$ is an i.i.d.\ zero mean Gaussian random vector
independent of $\wbf^p$ and $\tilde{\qbf}_j$, $j=0,\ldots,\km1$.
\end{lemma}
\begin{proof}
This is a direct application of Lemma~\ref{lem:orthogGaussLim}.
Let $\xbf = \Ubf_{\Abf_k^\perp}\tran\vbf_k$ so that
\[
    \qbf_k^{\rm det} = \Ubf_{\Bbf_k^\perp}\Vbf\tran\xbf_k.
\]
For each $N$, $\Ubf_{\Bbf_k^\perp} \in \R^{N \x (N-s)}$ is a matrix
with orthonormal columns spanning $\Range(\Bbf_k)^\perp$.
Also, since $\tilde{\Vbf}$ is uniformly distributed on the set of
$(N-s)\x (N-s)$ orthogonal matrices, and independent of $G_k$,
it is independent of $\xbf$.
Lemma~\ref{lem:rhoconv} also shows that
\[
    \lim_{N \arr \infty} \frac{1}{N} \|\xbf\|^2 = \rho_k,
\]
almost surely.
The limit \eqref{eq:qranlim} now follows from
Lemma~\ref{lem:orthogGaussLim}.
\end{proof}

Using the partition \eqref{eq:qpart} and Lemmas~\ref{lem:qconvdet} and \ref{lem:qconvran},
we have that
\[
    \qbf_k = \tilde{\qbf}_k + \ONinv, \quad
     \tilde{\qbf}_k := \beta_{k0}\tilde{\qbf}_0 +
        \cdots + \beta_{k,\km1}\tilde{\qbf}_{\km1} + \ubf.
\]
Now, by the induction by hypothesis, the matrix $\tilde{\Qbf}_{\km1}$ are
have i.i.d.\ rows that are jointly Gaussian.   The matrix $\tilde{\Qbf}_k$ is formed by
adding the column $\tilde{\qbf}_k$ to $\tilde{\Qbf}_{\km1}$.
Since $\ubf$ is Gaussian i.i.d.\ independent of $\tilde{\qbf}_j$ for $j < k$,
we have that the matrix $\tilde{\Qbf}_k$ will have i.i.d.\ rows that are jointly Gaussian.

It remains to show all the limits in \eqref{eq:ag2limgen}.
First,
\begin{align*}
    E[\tilde{q}_{nk}^2] &\stackrel{(a)}{=} \lim_{N \arr \infty} \frac{1}{N} \|\tilde{\qbf}_k\|^2
    \\
    &\stackrel{(b)}{=} \lim_{N \arr \infty} \frac{1}{N} \|\qbf_k\|^2
    \stackrel{(d)}{=}
    \lim_{N \arr \infty} \frac{1}{N} \|\vbf_k\|^2 \stackrel{(d)}{=} \tau_{2k},
\end{align*}
where (a) follows from the Strong Law of Large Numbers and the fact that the components of
$\tilde{\qbf}_k$ are i.i.d.;
(b) follows from the fact that $\qbf_k = \tilde{\qbf}_k + \ONinv$;
(c) follows from \eqref{eq:qupgen} and the fact that $\Vbf$ is orthogonal;
and (d) follows from Lemma~\ref{lem:vpcorr}.
Now the function $\Gamma_1(\gamma_1,\alpha_1)$ is assumed to be
continuous at $(\gammabar_{1k},\alphabar_{1k})$.  Also, the induction hypothesis assumes that
$\alpha_{1k} \arr \alphabar_{1k}$ and $\gamma_{1k} \arr \gammabar_{1k}$ almost surely.
Hence,
\beq \label{eq:gam2limpf}
    \lim_{N \arr \infty} \gamma_{2k} = \lim_{N \arr \infty} \Gamma_1(\gamma_{1k},\alpha_{1k})
    = \gammabar_{2k}.
\eeq
In addition,
\begin{align}
    \lim_{N \arr \infty} \alpha_{2k}
    &\stackrel{(a)}{=} \lim_{N \arr \infty} \bkt{\nabla \fbf_q(\qbf_k,\wbf^q,\gamma_{2k})} \nonumber \\
    &\stackrel{(b)}{=} \lim_{N \arr \infty} \bkt{\nabla \fbf_q(\tilde{\qbf}_k,\wbf^q,\gammabar_{2k})}
    \stackrel{(c)}{=} {\mathcal A}_q(\tau_{2k},\gammabar_{2k})
    \stackrel{(d)}{=} \alphabar_{2k}, \label{eq:a2limpf}
\end{align}
where (a) follows from \eqref{eq:alpha2gen};
(b) follows from the Lipschitz continuity assumptions of $\fbf_q(\cdot)$;
(c) follows from \eqref{eq:divpq}
and (d) follows from \eqref{eq:a2segen}.
The limits \eqref{eq:gam2limpf} and \eqref{eq:a2limpf} prove \eqref{eq:ag2limgen}.
This completes the induction argument and the proof of the theorem.

\section{Proof of Theorem~\ref{thm:se} } \label{sec:sepf}

The proof is virtually identical to that used in \cite{rangan2016vamp}.
Specifically, we show that
Theorem~\ref{thm:se} is a special case of Theorem~\ref{thm:genConv}.
As in \cite{rangan2016vamp},
we need to simply rewrite the recursions in Algorithm~\ref{algo:vamp} in the form
\eqref{eq:algoGen} by defining the error terms
\beq \label{eq:pvslr}
    \pbf_k := \rbf_{1k}-\xbf^0, \quad
    \vbf_k := \rbf_{2k}-\xbf^0,
\eeq
and their transforms,
\beq \label{eq:uqslr}
    \ubf_k := \Vbf\tran\pbf_k, \quad
    \qbf_k := \Vbf\tran\vbf_k.
\eeq
Also, define the disturbance terms
\beq \label{eq:wpqslr}
    \wbf^q := (\xibf,\sbf), \quad
    \wbf^p := \xbf^0, \quad \xibf := \Ubf\tran\wbf.
\eeq
Also, define the update functions,
\begin{subequations} \label{eq:fqpslr}
\begin{align}
    \fbf_q(\qbf,(\xibf,\sbf),\gamma_2) &:= \frac{\gamma_w \sbf\xibf + \gamma_2 \qbf}{
        \gamma_w \sbf^2 + \gamma_2}, \label{eq:fqslr} \\
    \fbf_p(\pbf,\xbf^0,\gamma_1) &:= \gbf_1(\pbf+\xbf^0,\gamma_1) - \xbf^0.
        \label{eq:fpslr}
\end{align}
\end{subequations}
In the definition of the function $\fbf_q(\cdot)$, the product $\sbf\xibf$ and the division
are to be taken componentwise.
Also, let
\[
    C_i(\alpha_i) := \frac{1}{1-\alpha_i}, \quad \Gamma_i(\gamma_i,\alpha_i)
    :=  \gamma_i\left[\frac{1}{\alpha_i}-1 \right].
\]
Then, it is shown in \cite{rangan2016vamp} that the
recursions in Algorithm~\ref{algo:vamp} exactly match \eqref{eq:algoGen}.

So, all we need to do is show that the update functions in \eqref{eq:fqpslr}
satisfy Definitions~ \ref{def:uniflipgen} and \ref{def:convgen}.
These conditions are proven in the next two lemmas.
By the assumption of Theorem~\ref{thm:se}, $\gammabar_{2k} > 0$ for all $k$.
So, for any finite $k$, there exists a lower bound $\gamma_{2,min} > 0$
such that $\gammabar_{2\ell} \geq \gamma_{2,min}$ for all $\ell \leq k$.
Let $G_2 = \{\gamma_2 | \gamma_2 \geq \gamma_{2,min}\}$.

\begin{lemma} \label{lem:fqas}
The sequence of random vectors $\wbf^q$ in \eqref{eq:wpqslr},
functions $\fbf_q(\cdot)$ in \eqref{eq:fqslr} satisfy Definitions~\ref{def:uniflipgen}
and \ref{def:convgen} for $\gamma_2 \in G_2$.
\end{lemma}
\begin{proof}
First note that the function $\fbf_q(\cdot)$ in \eqref{eq:fqslr}
is separable meaning that its $n$-th output is given by,
\beq \label{eq:fqsep}
    \left[ \fbf_q(\qbf,(\xibf,\sbf),\gamma_2)\right]_n
    = \phi(q,s,\xi,\gamma_2) := \frac{\gamma_w s\xi + \gamma_2 q}{
        \gamma_w s + \gamma_2}.
\eeq
For any $\gamma_2 \in G_2$, we can bound the partial derivatives,
\begin{align*}
    \left|\frac{\partial \phi(q,s,\xi,\gamma_2)}{\partial q}\right| &=
        \left|\frac{\gamma_2}{\gamma_w s + \gamma_2}\right| \leq 1, \\
    \left|\frac{\partial \phi(q,s,\xi,\gamma_2)}{\partial \gamma_2}\right| &=
        \left|\frac{q(\gamma_w s + \gamma_2) - \gamma_ws\xi-\gamma_2 q}{(\gamma_w s + \gamma_2)^2}\right|
    \\
    &\leq [|q|+|\xi|] \frac{\gamma_w s}{(\gamma_w s + \gamma_2)^2}
        \leq [|q|+|\xi|] \frac{1}{\gamma_{2,min}^2}.
\end{align*}
Therefore, if we let $A=1$, $B=C=1/\gamma_{2,min}^2$, we get that,
\[
    |\phi(q_2,s,\xi,\gamma_{22}) - \phi(q_1,s,\xi,\gamma_{21})|
    \leq (A+B|\gamma_{22}-\gamma_{21}|)|q_2-q_1| + C|\xi||\gamma_{22}-\gamma_{21}|,
\]
for and $q_1,q_2$ and $\gamma_{21},\gamma_{22} \in G_2$.
This implies that for any vectors $\qbf_1,\qbf_2$,
\begin{align*}
    \MoveEqLeft \frac{1}{\sqrt{N}}
    \| \fbf_q(\qbf_2,(\xibf,\sbf),\gamma_{22}) -
        \fbf_q(\qbf_1,(\xibf,\sbf),\gamma_{21})\| \\
    & \leq \frac{(A+B|\gamma_{22}-\gamma_{21}|)}{\sqrt{N}}\|\qbf_2-\qbf_1\| +
    C\frac{\|\xibf\|}{\sqrt{N}}|\gamma_{22}-\gamma_{21}|.
\end{align*}
Since $\xibf := \Ubf\tran\wbf$ and $\Ubf$
is orthogonal, $\|\xibf\|=\|\wbf\|$.  Also, since $\wbf \sim \Norm(\mathbf{0},\Ibf/\gamma_w)$,
\[
    \lim_{N \arr \infty} \frac{1}{N} \|\xibf\|^2 = \lim_{N \arr \infty} \frac{1}{N} \|\wbf\|^2
    = \frac{1}{\gamma_w},
\]
almost surely.
Therefore,
\begin{align*}
    \MoveEqLeft \limsup_{N \arr \infty}
    \frac{1}{\sqrt{N}}
    \| \fbf_q(\qbf_2,(\xibf,\sbf),\gamma_{22}) -
        \fbf_q(\qbf_1,(\xibf,\sbf),\gamma_{21})\| \nonumber \\
    & \leq
    \limsup_{N \arr \infty}
    \frac{(A+B|\gamma_{22}-\gamma_{21}|)}{\sqrt{N}}\|\qbf_2-\qbf_1\| + \frac{C}{\sqrt{\gamma_w}}
    |\gamma_{22}-\gamma_{21}|,
\end{align*}
which proves that $\fbf_q(\cdot)$ satisfies the uniform Lipschitz condition in
Definition~\ref{def:uniflipgen}.

We turn to the convergence properties in Definition~\ref{def:convgen}.
For each $N$, let $\qbf_1,\qbf_2$ be vectors with components
$(q_{1n},q_{2n})$ that are i.i.d.\ and Gaussian
$(q_{1n},q_{2n}) \sim \Norm(0,\Sbf)$ for some positive definite covariance matrix $\Sbf$.
Let $\gamma_{21},\gamma_{22} > 0$.
Since $\xibf := \Ubf\tran\wbf$, $\Ubf$
is orthogonal, and $\wbf \sim \Norm(\mathbf{0},\Ibf/\gamma_w)$,  we have that
$\xibf \sim \Norm(\mathbf{0},\Ibf/\gamma_w)$.
Hence, the components of $\xibf$ are i.i.d.
Also, by assumption, $\sbf$ has i.i.d.\ components,
independent of $\xibf$.  Therefore,
\begin{align*}
    \MoveEqLeft \lim_{N \arr \infty} \frac{1}{N}
    \fbf_q(\qbf_2,(\xibf,\sbf),\gamma_{22})\tran
        \fbf_q(\qbf_1,(\xibf,\sbf),\gamma_{21}) \\
    & \stackrel{(a)}{=}
    \lim_{N \arr \infty} \frac{1}{N}
    \phi(q_{2n},\xi_n,s_n,\gamma_{22})\phi(q_{1n},\xi_n,s_n,\gamma_{21}) \nonumber \\
    & \stackrel{(b)}{=}
    \lim_{N \arr \infty} \Exp\left[
    \phi(q_{2n},\xi_n,s_n,\gamma_{22})\phi(q_{1n},\xi_n,s_n,\gamma_{21}) \right],
\end{align*}
where (a) follows from the separability of $\fbf_q(\cdot)$ in \eqref{eq:fqsep}
and (b) follows from the fact that terms are i.i.d., so we can apply the Strong Law
of Large Numbers.  The convergence of the limit is almost sure.
This proves \eqref{eq:momlimgen}.  The limit \eqref{eq:divlimgen} can be proven
similarly.  Hence, the sequences $\wbf^q$ and $\fbf_q(\cdot)$ satisfy Definition~\ref{def:convgen}.
\end{proof}

Next, consider $\fbf_p(\cdot)$ in \eqref{eq:fpslr}.

\begin{lemma} \label{lem:fpas}
The sequence of random vectors $\wbf^p$ in \eqref{eq:wpqslr},
functions $\fbf_p(\cdot)$ in \eqref{eq:fpslr} satisfy Definitions~\ref{def:uniflipgen}
and \ref{def:convgen}.
\end{lemma}
\begin{proof}
For any vectors $\pbf_1$, $\pbf_2$ and $\gamma_1$, $\gamma_2$,
\begin{align*}
  \MoveEqLeft
  \|\fbf_p(\pbf_2,\xbf^0,\gamma_2)- \fbf_p(\pbf_1,\xbf^0,\gamma_1)\|
  = \|\gbf_1(\pbf_2+\xbf^0,\gamma_2)- \gbf_1(\pbf_1+\xbf^0,\gamma_1)\| \nonumber \\
  &\leq (A + B|\gamma_2-\gamma_1|)\|\pbf_2-\pbf_1\|^2 + \sqrt{N}C|\gamma_2-\gamma_1|,
\end{align*}
where the last step follows from the fact that $\gbf_1(\cdot)$
is uniformly Lipschitz continuous as per Definition~\ref{def:uniflip}.
This shows that $\fbf_p(\cdot)$ satisfies the uniform Lipschitz continuity assumption
in Definition~\ref{def:uniflipgen}.

Now suppose that $\pbf_1,\pbf_2$ are Gaussian vectors such that the components,
$(p_{1n},p_{2n})$ are i.i.d.\ with $(p_{1n},p_{2n}) \sim \Norm(0,\Sbf)$.
Then,
\begin{align*}
    \MoveEqLeft
    \lim_{N \arr \infty} \frac{1}{N} \fbf_p(\pbf_1,\xbf^0,\gamma_1)\tran
        \fbf_p(\pbf_2,\xbf^0,\gamma_2) = \\
    &=\lim_{N \arr \infty} \frac{1}{N} \left[ \gbf_1(\pbf_1+\xbf^0,\gamma_1)\tran
        \gbf_1(\pbf_2+\xbf^0,\gamma_2)
        - 2(\xbf^0)\tran \gbf_1(\pbf_1+\xbf^0,\gamma_1) + \|\xbf^0\|^2 \right].
\end{align*}
All three terms on the right-hand side of this equation converge due to
the assumption that the limits in~\eqref{eq:convlim} converge.
Moreover, the limits are continuous in $\Sbf$, $\gamma_1$ and $\gamma_2$.
The convergence of \eqref{eq:divlimgen} can be proven similarly.
Hence, the sequences $\wbf^p$ and $\fbf_p(\cdot)$ satisfy Definition~\ref{def:convgen}.
\end{proof}

\medskip
Lemmas~\ref{lem:fqas} and \ref{lem:fpas} show that the vectors $\wbf^q$ and $\wbf^p$
and functions $\fbf_q(\cdot)$ and $\fbf_p(\cdot)$ satisfy the necessary conditions
of Theorem~\ref{thm:genConv}, which completes the proof of Theorem~\ref{thm:se}.

\section{Example Image Recoveries} \label{sec:image}

Figure~\ref{fig:images} shows the original images and examples of recovered
images for various algorithms after 12 iterations under
sampling rate $M/N=0.3$, $\text{cond}(\Abf)=1$, and no noise.
There we see that the quality of DnCNN-based recovery far exceeds that of LASSO.
The figure also shows that, in all cases, LASSO-VAMP outperformed LASSO-AMP 
and that in all but one case DnCNN-VAMP outperformed DnCNN-AMP.

\begin{figure}[t]
  \includegraphics[height=\textwidth,
                   angle=90,origin=cB,
                   clip=false]{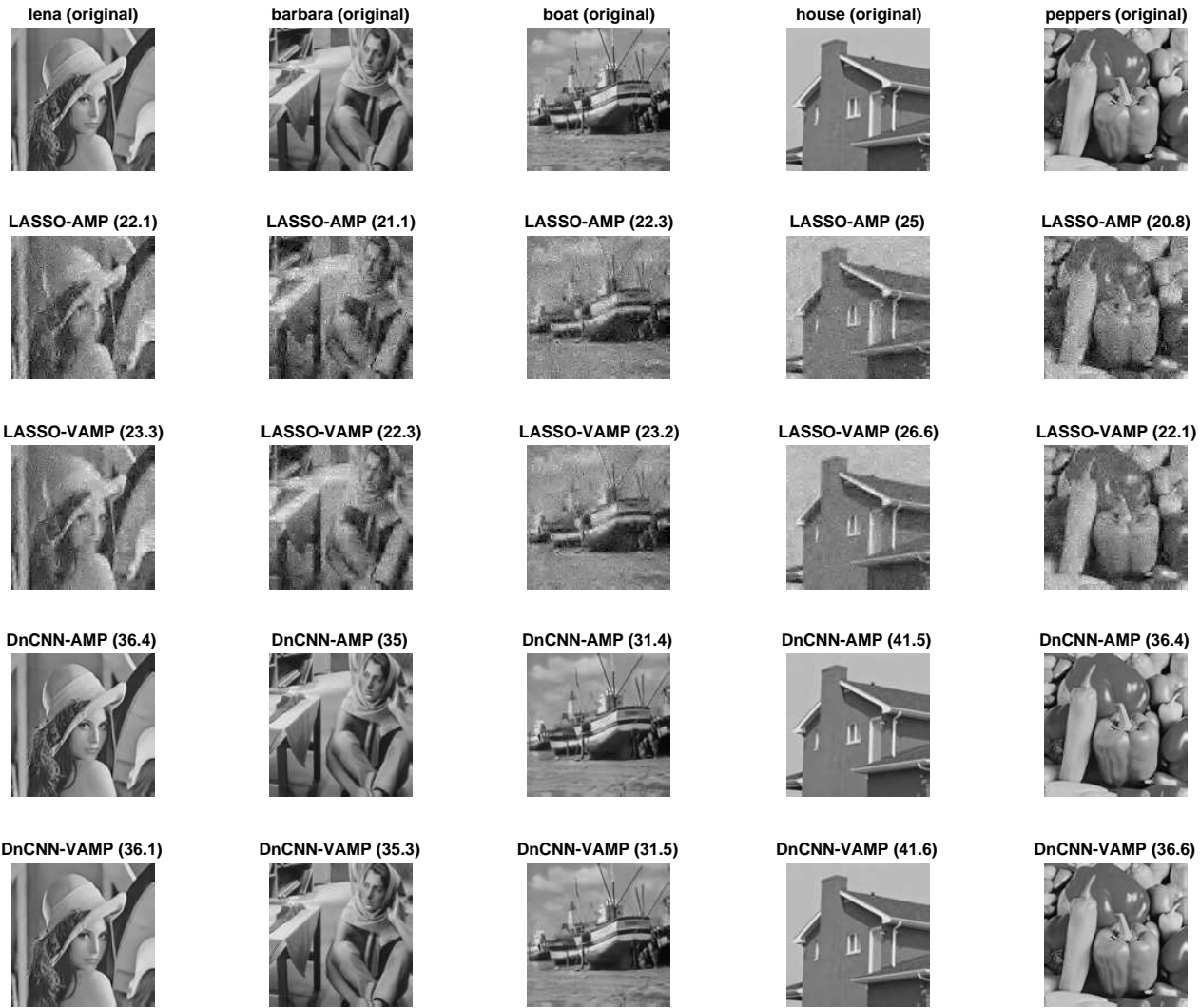}
  \caption{Compressive image recovery at $M/N=0.3$: Original and recovered images (with PSNR)}
  \label{fig:images}
\end{figure}

\end{document}